\documentclass[fleqn,10pt]{wlscirep}
 
\usepackage{graphicx}
\usepackage{amsmath}
\usepackage{amssymb}
\usepackage{amsthm}
\usepackage{hyphenat}
\usepackage{makeidx}
\usepackage{multirow}
\usepackage{mathtools}
\usepackage{caption}
\usepackage{subcaption}
\usepackage{rotating}
\usepackage{wrapfig}
\DeclareMathOperator{\sech}{sech}
\DeclareMathOperator{\csch}{csch}

\newtheorem{theo}{Theorem}
\newtheorem{definit}{Definition}

\usepackage{ulem}
\definecolor{pdcolor}{rgb}{1,0.5,0}

\definecolor{pdblue}{rgb}{0,0,1}

\title{Cellular automaton models for time-correlated random walks: derivation and analysis}
\author[1]{J. M. Nava-Sede\~no}
\author[1,2]{H. Hatzikirou}
\author[3]{R. Klages}
\author[1]{A. Deutsch}

\affil[1]{Technische Universit\"at Dresden, Center for Information Services and High Performance Computing, N\"othnitzer Stra{\ss}e 46, 01062 Dresden, Germany}
\affil[2]{Department of Systems Immunology and Braunschweig Integrated Centre of Systems Biology, Helmholtz Center for Infection Research, Inhoffenstra{\ss}e 7, 38124 Braunschweig, Germany}
\affil[3]{School of Mathematical Sciences, Queen Mary University of London, Mile End Road, London E1 4NS, United Kingdom}

\begin{abstract}

   Many diffusion processes in nature and society were found to be
  anomalous, in the sense of being fundamentally different from
  conventional Brownian motion. An important example is the migration
  of biological cells, which exhibits non-trivial temporal decay of
  velocity autocorrelation functions. This means that the
  corresponding dynamics is characterized by memory effects that
  slowly decay in time. Motivated by this we construct non-Markovian
  lattice-gas cellular automata models for moving agents with memory.
  For this purpose the reorientation probabilities
   are derived from velocity autocorrelation functions that are
  given a priori; in that respect our approach is
  ``data-driven''. Particular examples we consider are velocity
  correlations that decay exponentially or as power laws, where the
  latter functions generate anomalous diffusion. The computational
  efficiency of cellular automata combined with our analytical results
  paves the way to explore the relevance of memory and anomalous
  diffusion for the dynamics of interacting cell populations, like
  confluent cell monolayers and cell clustering.

\end{abstract}

\begin{document}
\maketitle

\section*{Introduction}

Within the past two decades transport processes in many branches of
the sciences were observed to be {\it anomalous}, in the sense that
they do not obey the laws of conventional statistical physics and
thermodynamics
\cite{MeKl00,CKW04,MeKl04,KRS08,KlSo11,HoFr13,MJJCB14,ZDK15}. Important
cases are diffusion processes where the long-time mean square
displacement (MSD) does not grow linearly in time. That is,
$\left\langle r^2 \right\rangle \propto t^{\phi}$, where the angular
brackets denote an ensemble average, does not increase with $\phi=1$
as expected for Brownian motion but either {\it subdiffusively} with
$\phi<1$ or {\it superdiffusively} with $\phi>1$
\cite{SZK93,KSZ96,SKB02}. After pioneering work on amorphous
semiconductors \cite{SM75}, more recently anomalous diffusion has been
detected in many other complex systems
\cite{MeKl04,KRS08,HoFr13,ZDK15}; here well-known examples of physcial
systems are nanopores \cite{Kuk96}, plasmas \cite{Bale05} and glassy
material \cite{BBW08}.\

 Biological systems frequently exhibit anomalous
properties as well: Prominent examples are the foraging of organisms
\cite{VLRS11}, epidemic spreading \cite{BHG06} and the diffusion of
macromolecules in biological cells \cite{HoFr13}. Especially, it was
found that many types of cells migrate anomalously: {\it Hydra} cells
\cite{URGS01}, mammary gland epithelial cells \cite{TWHSH03}, {\it
  MDCKF} cells \cite{DDPKS03}, amoeboid {\it Dictyostelium} cells
\cite{TSYU08,BBFB10}, T cells \cite{HaBa12}, breast carcinoma cells
\cite{MMSLSF15} and stem cells \cite{BBG14} were all experimentally
observed to move superdiffusively, typically with non-Gaussian
position and/or velocity distribution functions
\cite{URGS01,TWHSH03,DDPKS03,TSYU08,BBFB10,HaBa12,MMSLSF15}
accompanied by either exponential or non-exponential position
\cite{HaBa12}, and exponential \cite{RUGOS00,TSYU08,BBFB10} or power
law \cite{URGS01,TWHSH03,DDPKS03} velocity autocorrelation function
(VACF) decay. For T cells it was argued that superdiffusion optimizes
their search to kill intruding pathogens \cite{HaBa12,KBG16}. While all these results are on single cell migration, currently
collective cell migration is moving into the center of interest
\cite{MeVi14}, where cells interact with each other, e.g., by chemical
signalling \cite{McKPL10}. Interesting phase transitions inside dense
tissues of epithelial cell monolayers were reported \cite{Malin17} and
partially traced back to particular features of single-cell migration
\cite{BYMM16}. It was also observed experimentally that superdiffusion
appears to foster the formation of clusters of stem cells leading to
tissue formation \cite{BBG14}. Other works investigate the role of
interacting agents for phase transitions in active matter
\cite{FNCTVW16}, and collective anomalous dynamics emerging from the
interaction of single agents \cite{FeKo17,ARBPHB15}.

On the theoretical side there are many different ways to model
anomalous diffusion in terms of stochastic processes, such as
continuous time random walks (CTRW) \cite{MeKl00,KlSo11}, generalized
Langevin equations \cite{CKW04}, L\'evy flights and walks
\cite{ZDK15}, fractional diffusion equations
\cite{MeKl00,KRS08,KlSo11}, scaled Brownian motion and heterogeneous
diffusion processes \cite{MJJCB14}. A subset of these models, most
notably generalized Fokker-Planck equations
\cite{URGS01,TWHSH03,DDPKS03}, generalized Langevin equations
\cite{TSYU08,BBFB10} and generalized random walks
\cite{HaBa12,MMSLSF15}, has been used to model anomalous movement of single
biological cells. However, solving equations for anomalously moving
single particles analytically or numerically is typically difficult
already \cite{MeKl00,CKW04,KRS08,KlSo11,MJJCB14,ZDK15}. 
To our knowledge there exists no systematic attempt to generalize this
theory to model interacting many-particle systems; the only exception
we are aware of is a line of work in plasma physics \cite{Bale05}.

On the other hand, several models have been introduced to study the
collective movement of particles and cells \cite{lgcacarsten,colloids,rodsw} . Cellular automata (CA) in particular have the advantage
of being less computationally demanding than continuous models when
performing simulations.  A specific type of CA is the so-called
lattice-gas cellular automata (LGCA) \cite{lgca}. In LGCA,
each lattice node can contain several particles, which at each time
step are rearranged within the lattice node according to the
interaction rule, and subsequently moved to a neighboring node.  In a
biological context particles can be regarded as cells, while the LGCA
rules mimic cell migration and interaction.
Furthermore, LGCA have proved to be amenable to mathematical analysis \cite{mf-lgca}. 
 For this reason, LGCA
have been introduced as mesoscopic models for single and collective
cell migration \cite{heterog,allee,glioma,lgcamig}. So far, none of the mentioned models has considered
anomalous migration of single cells.

 It
thus arises the need to design simple fundamental schemes by which the
collective properties of interacting agents can be studied whose
individual dynamics is anomalous. Using our methods will enable to explore the
relevance of microscopic single-particle dynamics for emerging
collective phenomena. We thus devise a scheme by which anomalous
dynamics of many interacting agents can be simulated efficiently,
which is based on capturing the non-trivial decay of VACFs. This
approach generates superdiffusion if the correlation decay is of power
law-type \cite{URGS01,TWHSH03,DDPKS03}. We emphasize that our
data-driven approach can be applied to any moving entity that exhibits
dynamics with non-trivial correlation decay, a feature that may be
expected to hold more generally for the movement of biological
organisms \cite{VLRS11}.

We use the LGCA modeling framework and construct various time-correlated random walk models. After briefly
introducing the LGCA concept, we define an LGCA model for unbiased random walk.
Next, motivated by the biophysical mechanism of single cell crawling
we construct a {\it persistent random walk}
LGCA model wherein angular (orientation) correlations give rise to temporal
correlations. Subsequently, we construct a first LGCA model for {\it
  time-correlated random walk} which is data-driven, as the model's
reorientation probabilities are derived by assuming that the exact
temporal dependence of the VACF is known a priori. Finally, we develop
a {\it generalized time-correlated random walk} LGCA model for
cell movement at short and medium time regimes by curing a deficiency
of our first time-correlated random walk model for short times. Figure \ref{models} shows single
cell tracks with the corresponding VACFs and MSDs for our main two
classes of LGCA models we are dealing with, which are Markovian and non-Markovian random walks, exemplified
by showing their basic features.

\begin{figure}
\begin{center}
\includegraphics[width=125mm]{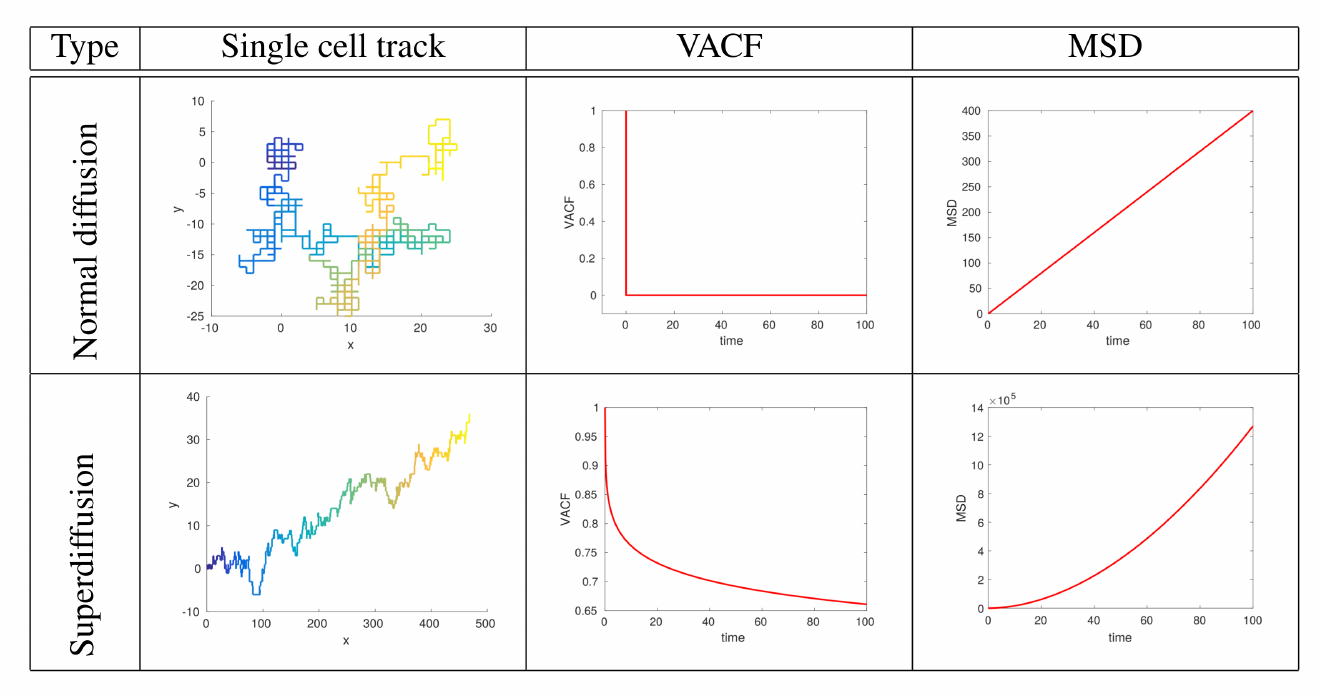}
\end{center}
 \caption{Basic types of diffusive movement in two dimensions characterized by two key quantities (VACF and MSD). Shown in the first column are the tracks of a single particle starting at $(x,y)=(0,0)$ that exhibits either long-time normal diffusion (top row) or superdiffusion (bottom row). The color gradient changes from blue to yellow with elapsed time. The second column displays the particle's velocity autocorrelation function (VACF) Eq.~(\ref{lgca-corr}), the third column its mean square displacement (MSD) Eq.~(\ref{eq:msd}).}
\label{models}
\end{figure}

\section*{Lattice-gas cellular automata}

Cellular automata are mathematical models where the states of discrete
lattice nodes are updated at discrete time
steps. If the states of the lattice sites are Boolean, such states
can be interpreted as presence/abscence of a particle at a particular node.
 The lattice-gas cellular automaton is a specific CA type, which has two important characteristics:
first, particle reorientation and migration are separated into
a probabilistic and a deterministic step, respectively. Secondly, 
to each node, $b$ velocity channels are associated which can be occupied by at most one particle (exclusion principle).
The set of velocity channels is given by
$\vec{c}_j = \left( \cos \frac{2\pi j}{b},\sin \frac{2\pi j}{b}
\right)$, $i\in\left\{ 0,1,\ldots, b-1\right\}$ (see Fig \ref{lgca-dyn}).
  Particles move in discrete time steps of duration $\tau$ to
neighboring nodes located a distance $\varepsilon$
away in the lattice. At each time step
particles adopt the orientation
$\vec{c}_i$ with a probability $P_{i,k}$, called the reorientation
probability, where $t=k\tau$, $k\in\mathbb{N}$ is the elapsed
simulation time. Subsequently, cells will be deterministically
translocated to the nearest neighbor located in the direction of $\vec{c}_i$; see
Fig.~\ref{lgca-dyn} for a sketch of LGCA dynamics.
\begin{wrapfigure}{r}{0.4\textwidth}
\begin{center}
\includegraphics[width=70mm]{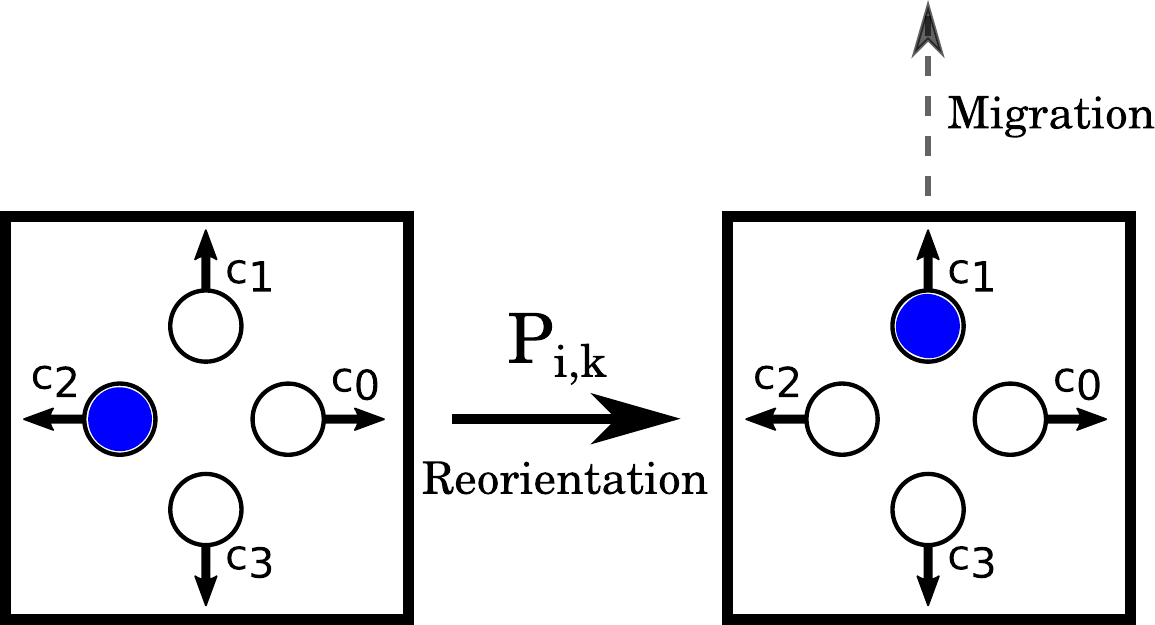}
\end{center}
\caption{LGCA dynamics. At each time
  step a particle is assigned an orientation $\vec{c}_i$ with a
  probability $P_{i,k}$. Subsequently, the particle is translocated to
  the nearest neighbor in the direction of its orientation.}
\label{lgca-dyn}
\end{wrapfigure}
If we were to have $N$ different particles in a single node, then the probability of 
the particles adopting the orientations $\vec{c}_{n^1},\vec{c}_{n^2},\ldots,\vec{c}_{n^N}$
would be given by 
\begin{equation}
P_{n^1,\ldots,n^N,k}=\prod_{\ell=1}^NP_{n^{\ell},k}.
\label{confprobs}
\end{equation}

\section*{Classical random walk}

Here we define an LGCA model for unbiased random walk\cite{Pea06}. 
In this model, at all timesteps, all orientations are chosen with equal probability \cite{rwtheo,mixing}. This means that the reorientation
probability is given by
\begin{equation}
 P_{i,k}=\frac{1}{b}.
 \label{rw-probs}
\end{equation}
In order to characterize the movement of a particle in this model we calculate
time-dependent expressions for the VACF and the MSD, which measure the
persistence, in terms of memory decay in time, and the spatial
exploratory power of a moving particle, respectively. In LGCA models
space, time and particle velocities are discrete so that the VACF is
given by \cite{KlKo02}
\begin{equation}
 g(k)=\left\langle \vec{c}_{i_0} \cdot \vec{c}_{i_k} \right\rangle=\sum_{i_k=1}^{b}P_{i_k,k}\left[ \vec{c}_{i_0} \cdot \vec{c}_{i_k}\right],
 \label{lgca-corr}
\end{equation}
where $\vec{c}_{i_k}$ is the orientation of the particle at the $k$-th
time step. The MSD is calculated by
\begin{equation}
 \left\langle r_k^2 \right\rangle =\sum_{r_k} r_k^2P_{r_k,k}, \label{eq:msd}
\end{equation}
where $r_k$ is the norm of the particle displacement at time step $k$
defined as $\vec{r}_k=\vec{x}_k-\vec{x}_0$, where $\vec{x}_k$ is the
position of the particle at time step $k$. 
The probability
$P_{r_k,k}$ can be calculated from Eq.~(\ref{rw-probs}) by noticing
that $\vec{r}_k=\varepsilon\sum_{k}\vec{c}_{i_k}$.
In this simple random walk model Eq.~(\ref{lgca-corr}) reduces to a
sum of cosines over homogeneously distributed angles. Hence the VACF is
given by $g(k)=\delta_{0,k}$, where $\delta_{i,j}$ is the
Kronecker delta. In the limit $\tau \rightarrow 0$ the VACF tends to
\begin{equation}
 g(t)=\delta(t),
 \label{rw-ct}
\end{equation}
where $\delta(t)$ is the Dirac delta function.  Eq.~(\ref{rw-ct})
means that the movement of the particle is uncorrelated as soon as it starts
moving, i.e.\ the orientation of the particle at any time step $k$ is
completely independent from its previous orientation, which is the
Markov property. On the other hand, simple combinatorics can be used
to calculate the particle's MSD yielding $\left\langle r_k^2
\right\rangle=k\varepsilon^2$. We can rewrite this expression by using
the general definition of the diffusion coefficient
\begin{equation}
 D=\lim_{t\to\infty}\frac{\left\langle r_t^2 \right\rangle}{2dt},
 \label{dc_def}
\end{equation}
where $d$ is the dimension of space. For a memoryless random walk this
equation boils down to $D_{rw}=\frac{\varepsilon^2}{2d\tau}$ with an
MSD of $\left\langle r_k^2 \right\rangle=2dD_{rw}k\tau$.  Given that
$\tau$ is the time step length and that $t=k\tau$ is the elapsed
time, the MSD is 
\begin{equation}
 \left\langle r_t^2 \right\rangle =2dD_{rw}t.
 \label{rw-msd}
\end{equation}
Eq.~(\ref{rw-msd}) shows that the classical random walk model
trivially yields a normal diffusion process, where the MSD increases
linearly in time \cite{Reif}.

\section*{Persistent random walk}

The assumption of the classical random walk in Eq.~(\ref{rw-probs})
that all the directions of movement are equally probable is generally
not true.  We will now construct an LGCA model
motivated by a simple biophysical model for a single, persistently
moving cell.

\subsection*{Rule derivation}

\begin{wrapfigure}{r}{0.35\textwidth}
\begin{center}
\includegraphics[width=50mm]{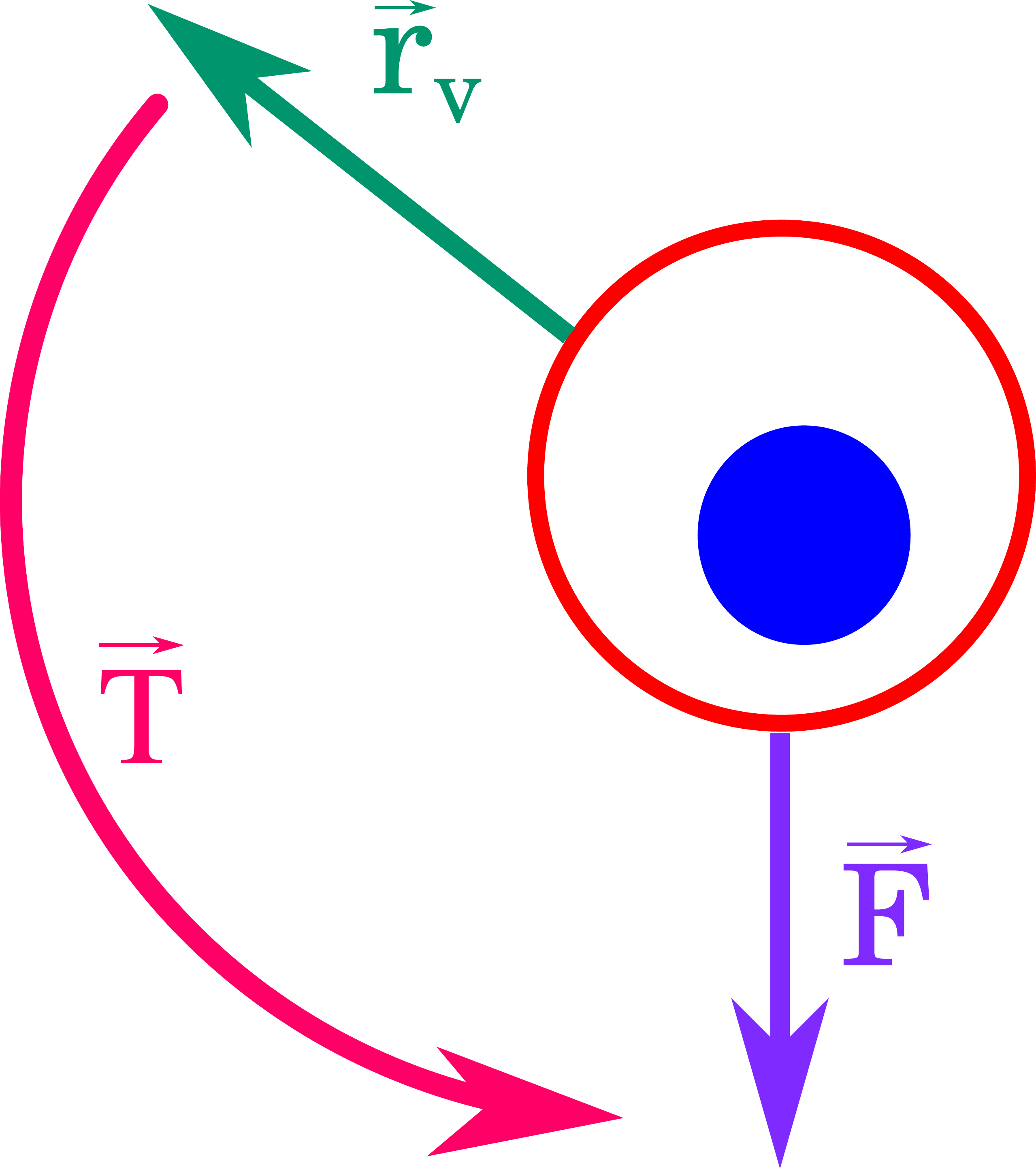}
\end{center}
\caption{Reorientation of a biological cell moving persistently. A cell which moves in a direction $\vec{r}_v$ feels an intracellular force $\vec{F}$ and reorients towards the direction of the force due to the torque $\vec{T}$.}
\label{per-model}
\end{wrapfigure}

Biological cells move by exerting forces to propel themselves. In the
case of eukaryotic cells such as fibroblasts, movement is achieved by
crawling over the substrate.  Crawling is performed by polymerization
of the actin cytoskeleton at the leading edge propelling the cell in
this direction.  We can identify the direction of the
actin concentration gradient with the direction of movement of a cell
$\vec{r}_v$. Furthermore, intracellular forces due to actin activity
$\vec{F}$ point towards the direction of new actin polymerization.
The cell then rotates/reorients/repolarizes due to the torque defined
by $\vec{T}=\vec{r}_v\times \vec{F}$, whose norm is given by
\begin{equation}
\label{normtorque}
 T=r_vF\sin\left(\theta_F-\theta_{r}\right),
\end{equation}
where $T=\| \vec{T} \|$, $r_v=\| \vec{r}_v \|$, $F=\| \vec{F} \|$,
and $\theta_r$ and $\theta_F$ are the angular components of
$\vec{r}_v$ and $\vec{F}$, respectively; see Fig.~\ref{per-model}. In
the overdamped regime, characteristic of the cellular environment, the
intracellular force $\vec{F}$ will not cause the cell to rotate
indefinitely but rather will cause the cell to rotate until
$\vec{r}_v$ and $\vec{F}$ are parallel.  Taking this into account, it
is possible to rewrite Eq.~(\ref{normtorque}) as
$T=r_vF\sin\left(\theta_{rt}-\theta_r\right)$, where $\theta_{rt}$ is
the direction of motion of the cell after it has finished its
reorientation.

The torque is given in terms of an energy of rotation $U(\theta_r,\theta_{rt})$ as
$\vec{T}=-\frac{\partial U(\theta_r,\theta_{rt})}{\partial
  \theta_r}$. Using this equation, the energy of rotation is then
given by
\begin{equation}
\label{energy}
U(\theta_r,\theta_{rt})=-\upsilon \cos\left(\theta_{rt}-\theta_r\right),
\end{equation}
where $\upsilon=r_vF$ is the amplitude of the torque generated inside
the cell. Having defined the energy of rotation, Eq.~(\ref{energy}), we can
describe the cell's reorientation by a Langevin equation \cite{peruani08} as $\frac{\partial \theta_r}{\partial
  t}=-\gamma\frac{U(\theta_r,\theta_{rt})}{\partial\theta_r}+\xi(t)$
with relaxation constant $\gamma$ and a zero-mean, delta correlated
noise term $\xi(t)$ such that $\left\langle \xi(t)
  \xi(t')\right\rangle = 2D_{\theta}\delta(t-t')$, where $D_{\theta}$
is the rotational diffusion coefficient. Based on this we can
immediately derive the LGCA reorientation probabilities
\cite{langlgca}. 

\begin{wrapfigure}{r}{0.4\textwidth}
\begin{center}
\includegraphics[width=70mm]{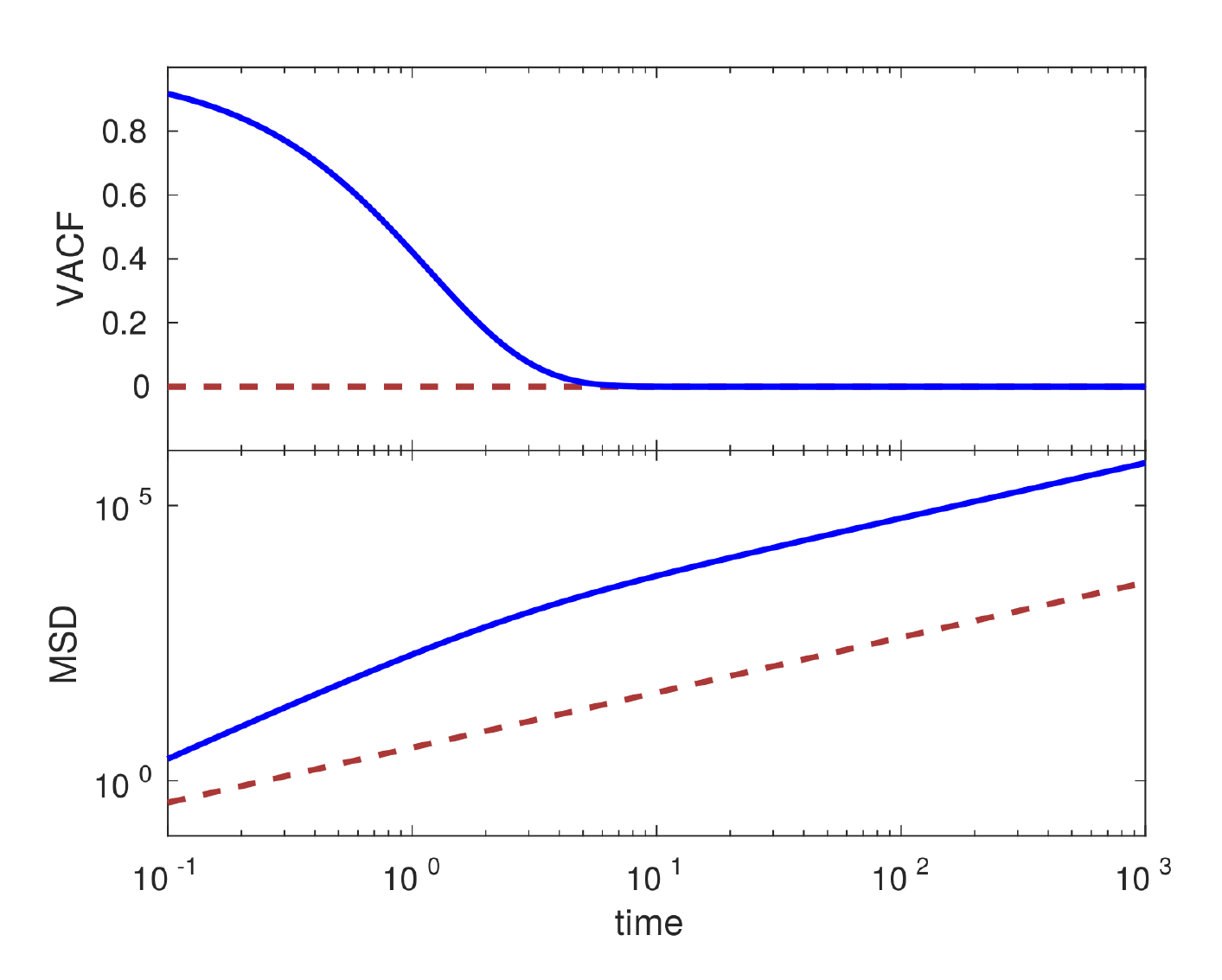}
\end{center}
\caption{Comparison between random walk and persistent random walk
  models. Shown are theoretical results for VACF (top) and MSD
  (bottom) in the random walk (dashed maroon line) and persistent (solid blue line)
  models. The parameter values are $D_{rw}=1$, $v=16$, and
  $\beta=5$.}
\label{vmrwmark}
\end{wrapfigure}

 These probabilities for a single cell then read
\begin{equation}
 P_{i_k,k}=\frac{1}{Z}\exp\left[ \beta \left( \vec{c}_{i_{k-1}} \cdot \vec{c}_{i_k} \right) \right],
 \label{lgcatrans}
\end{equation}
where $\vec{c}_{i_{k-1}}$ is the orientation of the cell at the
previous time step, $Z$ is the normalization constant (also known as the partition function) and
$\beta=\frac{\tilde{\gamma}}{D_{\theta}}$ is the sensitivity, where
$\tilde{\gamma}= \upsilon\gamma$ is the effective relaxation
constant.

\subsection*{Model analysis and results}
\label{vacfmarkov}

Using Eq.~(\ref{lgcatrans}) we can calculate the VACF and MSD for this
model.  By using the properties of the partition function $Z$ in
Eq.~(\ref{lgcatrans}) the VACF at every time step $k$ is (see Sec.~B in the
Supplementary Information)
\begin{equation}
 g(k)=\exp\left(\alpha k\right),
 \label{marktruec}
\end{equation}
where the exponent $\alpha$ depends heavily on the lattice dimension
and geometry. In particular, in a 2D square lattice we have
$\alpha=\ln \left[ \tanh \left( \frac{\beta}{2} \right)\right]$. In
all geometries the exponent is $\alpha <0$ over its domain $\beta>0$
(see again Sec.~B in the Supplementary Information).

Equation~(\ref{marktruec}) can be generalized to continuous time and
space by employing the relations between the time and space scalings,
namely the diffusion coefficient in the random walk limit $D_{rw}$ and
the instantaneous cell speed $v=\frac{\varepsilon}{\tau}$, where
$\varepsilon$ is the lattice spacing and $\tau$ the time step length.
Taking the limit $\tau \to 0$ yields the VACF in continuous time and
space
\begin{equation}
 g(t)=\exp\left(\frac{\alpha v^2}{2dD_{rw}} t\right).
 \label{markcontvacf}
\end{equation}

When time is discrete the MSD is given by \cite{Schu04,msdderiv} (see Sec.~A in the Supplementary Information) $\left\langle r^2
\right\rangle=2dD_{rw}k\tau+ \left\langle\sum_{i=1}^k \sum_{j=1}^k
  v^2\tau^2 \cos \left( \theta_i-\theta_j
  \right)\left(1-\delta_{ij}\right) \right\rangle$.  Calculating the
expected value on the right hand side we obtain (see again Sec.~A in
the Supplementary Information)
\begin{equation}
 \left\langle r_k^2 \right\rangle=2dD_{rw}k\tau+2v^2\sum_{i=1}^{k}\left(k\tau-i\tau \right)g(i)\tau.
 \label{corrmsd}
\end{equation}

Using Eq.~(\ref{marktruec}) and taking the limit of small time step
length we get
\begin{equation}
 \left\langle r_t^2 \right\rangle=2dD_{rw}t+2v^2\int_0^t(t-\tau)e^{\frac{v^2 \alpha}{2dD_{rw}}\tau}\mathrm{d}\tau,
 \label{markintmsd}
\end{equation}
which can be easily integrated to obtain the MSD for continuous time
\begin{equation}
 \left\langle r_t^2 \right\rangle=2dD_{rw}t\left(1-\frac{2}{\alpha}\right)+\left(\frac{2\sqrt{2}dD_{rw}}{v\alpha}\right)^2\left[\exp\left(\frac{v^2\alpha t}{2dD_{rw}}\right)-1\right].
 \label{markcontmsd}
\end{equation}

\begin{figure}
\centering
\includegraphics[width=150mm]{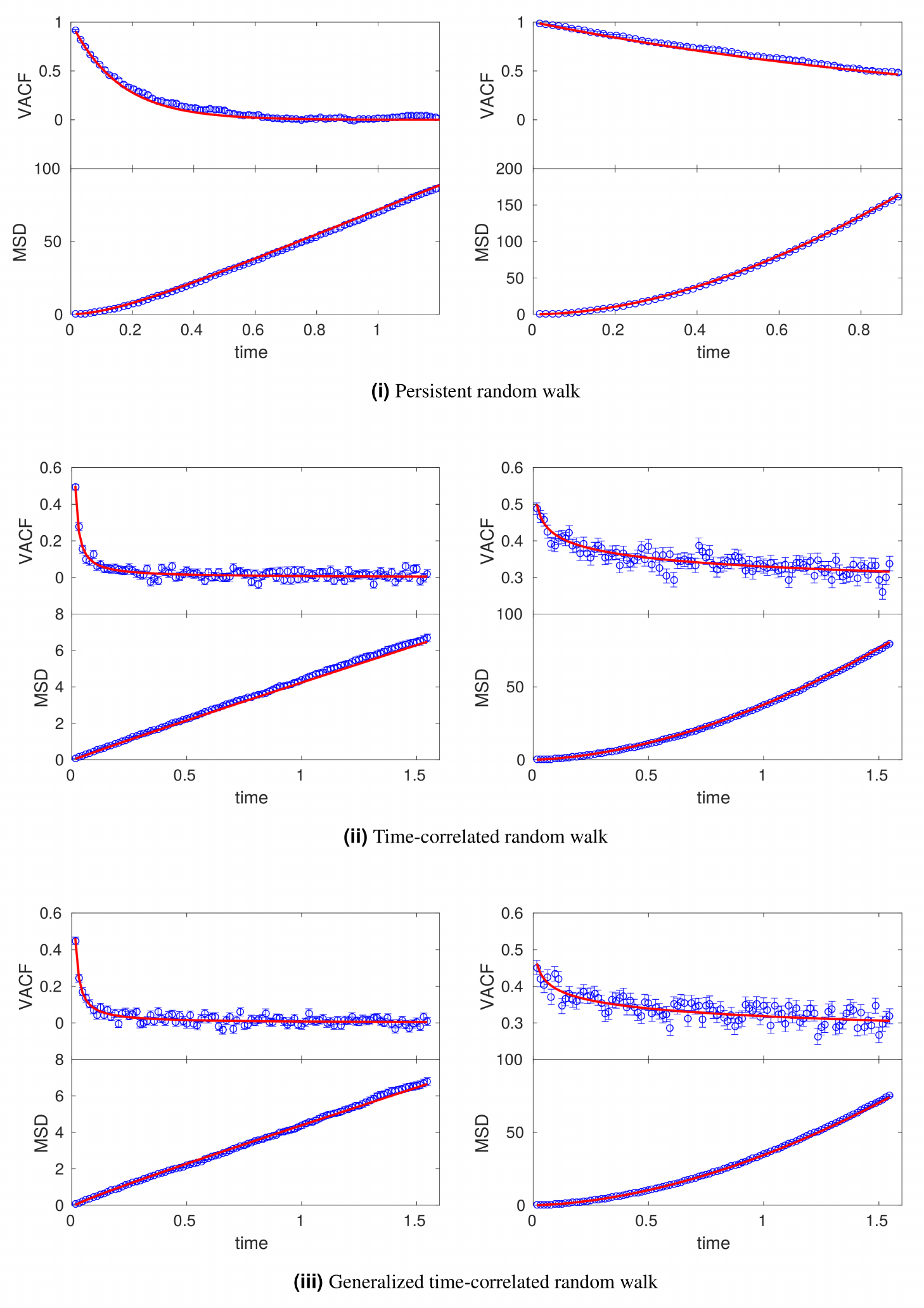}
\caption{Comparison between persistent, time-correlated and
  generalized time-correlated random walk models. Shown are simulation
  results (circles, mean $\pm$ standard error of the mean) and theoretical prediction (solid line) for VACF (top row) and MSD
  (bottom row). Parameters are $v=16$ and $D_{rw}=1$ in all cases. i)
  Sensitivity values: $\beta=3$ (left), $\beta=5$ (right). ii) and iii)
  $C_0=0.5$, $\Delta=0.016$, and exponents: $\phi=1$ (left), $\phi=0.1$ (right).}
\label{theovssims}
\end{figure}

Equation~(\ref{markintmsd}) agrees with the formal solution of the MSD
for an overdamped Langevin equation with colored noise
\cite{ChKl12}. Correspondingly, Eq.~(\ref{lgcatrans}) coincides with a
Langevin process where the noise is not white but colored whose
correlation is given by Eq.~(\ref{markcontvacf}). Using
Eq.~(\ref{markcontmsd}) we find that
$ \left\langle r^2 \right\rangle \propto t$ when $t \rightarrow
\infty$ rather quickly.

Comparing Eq.~(\ref{rw-ct}) to Eq.~(\ref{markcontvacf}) in
Fig.~\ref{vmrwmark} we see that in this second model the velocities
are no longer delta correlated but now decay exponentially in time.
On the other hand, for long times both Eqs.~(\ref{rw-msd}) and
(\ref{markcontmsd}) yield normal diffusion. However, for short times
cells performing persistent random walks move superdiffusively,
contrary to cells performing classical random walks; see again
Fig.~\ref{vmrwmark}. In Fig.~\ref{theovssims}i, Eqs.~(\ref{markcontvacf}) and
(\ref{markcontmsd}) are compared with results from LGCA simulations, where
 we see that the theory adequately predicts the observed simulation results.
Details on the computational implementation are found in Sec. H of the Supplementary Information.

\section*{Time-correlated random walk}

Due to the exponential decay of correlations the previous model did
not show superdiffusion at long time scales. It turns out that
finding a homogeneous, isotropic Markovian model that shows
superdiffusion and power-law decaying correlations is
not possible (for a proof see Sec.~C of the Supplementary Information).

\begin{theo}
  The velocity autocorrelation function of a particle whose
  orientations are given by a homogeneous, symmetric Markov chain is
  either delta-correlated, i.e. $g_k=\delta_{0,k}$, where $\delta$ is
  the Kronecker delta; alternating, i.e. $g_k=(-1)^ka^k$,
  $a\in\mathbb{R}^+$; or exponentially decaying, i.e.\ $g_k=e^{\alpha
    k}$, $\alpha\leq0$.
\end{theo}

To reproduce superdiffusion and power law decaying autocorrelations,
we will construct a non-homogeneous model by assuming that the time
dependency of the VACF is a known power law.

\subsection*{Rule derivation}

We now assume that the VACF $g(t)$ is known. In particular, if the
movement is power law-correlated the VACF has the form \cite{ChKl12}
\begin{equation}
 g(t)=C_0\left(\frac{\Delta}{t}\right)^{\phi}, \ t \geq \Delta,
 \label{anom_corr}
\end{equation}
where $\Delta > 0$ and $ 0 < \phi < \infty$ and assume that $\Delta\ll 1$,
to disregard the movement at short times, where Eq. \ref{anom_corr} diverges.
 The rate of decay
of the VACF is proportional to the exponent $\phi$. The crossover
time $\Delta$ specifies the time a which $g(t)=C_0$. The walk is positively
correlated if $C_0>0$, and anti-correlated if $C_0<0$.  Because the
process is non-homogeneous, $P_{i_k,k}$ in Eq.~(\ref{lgca-corr})
explicitly depends on the velocity channels $\vec{c_i}$ and the
current time step $k$. Combining Eqs.~(\ref{lgca-corr}) and
(\ref{anom_corr}) we obtain the following relation \cite{KlKo02}:
\begin{equation}
 \sum_{i=1}^{b}P_{i_k,k}\left[ \vec{c}_{i_0} \cdot \vec{c_{i_k}}\right]=g(k).
 \label{vacf_lgca}
\end{equation}

It is possible to derive the reorientation probabilities $P_{i_k,k}$ by
expanding Eq. \ref{vacf_lgca} for every time step.  Additionally, in
order to reduce the number of equations, we make the following
assumptions:

\begin{itemize}
\item The reorientation probabilities are independent, that is, the
  probability of following a certain trajectory is
  $P_{i_1,i_2,\cdots,i_k}=\prod_{j=1}^kP_{i_j,j}$.
\item There is symmetry around the initial orientation, i.e. if
  $\vec{c}_{i_k,k}\cdot \vec{c}_{i_0}=\vec{c}_{j_k,k}\cdot
  \vec{c}_{i_0}$ then $P_{i_k,k}=P_{j_k,k}$, $i_k\neq j_k$.
\end{itemize}
\label{1deriv}

\begin{wrapfigure}{r}{0.4\textwidth}
\begin{center}
\includegraphics[width=70mm]{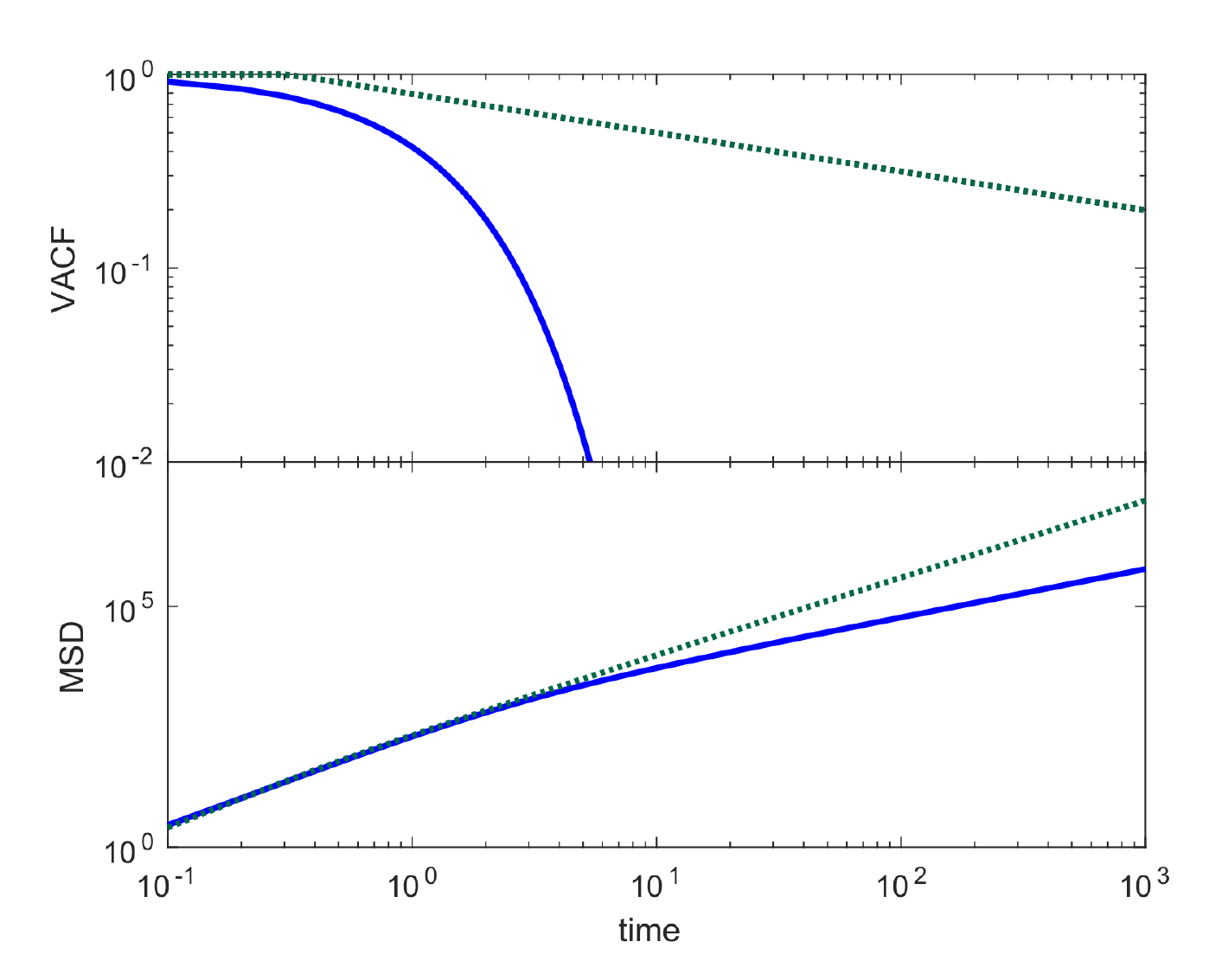}
\end{center}
\caption{Comparison between persistent and time-correlated random walk
  models. Shown are theoretical results for VACF (top) and MSD
  (bottom) in the persistent (solid blue line) and time-correlated (dotted green
  line) models. The parameter values are $D_{rw}=1$, $v=16$,
  $\beta=5$, $C_0=0.5$, $\Delta=10$, and $\phi=0.2$.}
\label{vmmarkpow}
\end{wrapfigure}

Using these assumptions we can derive the general expression for the
reorientation probabilities determining a certain VACF (see Sec.~D in the
Supplementary Information)
\begin{equation}
 P_{i_k,k}=\frac{1+d\left[ \vec{c}_{i_0} \cdot \vec{c}_{i_k} \right] g(k)}{b},
 \label{pxgeneral}
\end{equation}
where $d$ is the dimension of space and $b$ is the number of lattice
directions given by the lattice geometry.  If the VACF follows a power
law, then Eq.~(\ref{pxgeneral}) is always valid if the crossover time
$\Delta$ is smaller than the time step length, as the divergence of
Eq.~(\ref{anom_corr}) is avoided. If $\Delta \gg 0$ , we assume that 
the movement at short times is completely correlated, i.e. the VACF is given by
\begin{equation}
g(t)= \begin{cases} 1 & t\leq t^{\star} \\  C_0\left(\frac{\Delta}{t}\right)^{\phi} & t>t^{\star} \end{cases},
\end{equation}
where $t^{\star}$ is such that $
C_0\left(\frac{\Delta}{t^{\star}}\right)^{\phi}=1$.
We can
 then define a piecewise reorientation probability
\begin{equation}
P_{i_k,k}= \begin{cases} \delta\left(\vec{c}_{i_k}-\vec{c}_{i_0}\right) & k \leq \omega \\ \frac{1+d\left[ \vec{c}_{i_0} \cdot \vec{c}_{i_k} \right] g(k)}{b} & k>\omega\end{cases},
\label{pieceprobs}
\end{equation}
where $\omega\tau=t^\star$ is the duration of ballistic motion.

\subsection*{Model analysis and results}

For this model the VACF is obviously known as the reorientation
probabilities were calculated specifically to reproduce it. We now
calculate the MSD for this model.  The reorientation probabilities given
by Eq.~(\ref{pxgeneral}) only depend on the initial cell orientation
$\vec{c}_{i_0}$ and on the time step $k$, so they are independent of
other orientations at other times.  Because of this independence of
orientations the MSD is given by (see Sect.~A in the Supplementary Information)
\begin{equation}
\left\langle r_k^2 \right\rangle=2dD_{rw}\left[ k\tau -2\sum_{i=1}^k g^2(i)\tau\right] + 2v^2\sum_{i=1}^k \sum_{j=i}^kg(i)g(j)\tau^2
\label{msdind}
\end{equation}
yielding a Taylor-Green-Kubo formula \cite{KlKo02,KKCSG06}.  In
the limit of small time step lengths $\tau \rightarrow 0$ the MSD is
given by
\begin{equation}
 \left\langle r_t^2 \right\rangle=2dD_{rw}\left[ t - 2 \int_{\Delta}^t g^2(\tau)\mathrm{d}\tau \right]+ 2v^2 \int_{\Delta}^t \int_{\tau}^t g(\tau)g(k)\mathrm{d}k \mathrm{d}\tau .
 \label{msdgenpow}
\end{equation}
Equation~(\ref{msdgenpow}) shows that there is a correction to the random
walk diffusion coefficient as well as a new term depending on the
particle speed. If $g(t)$ is given by Eq.~(\ref{anom_corr}),
Eq.~(\ref{msdgenpow}) can be integrated yielding
\begin{subequations}
\begin{align}
 \begin{split}
  \left\langle r_t^2 \right\rangle=& 2dD_{rw}\left\{ t + \frac{2C_0^2\Delta}{1-2\phi} \left[1-\left( \frac{\Delta}{t} \right)^{2\phi-1}\right] \right\} \\ & +2\left(\frac{vC_0\Delta}{\phi-1}\right)^2\left[\frac{1}{2}\left(\frac{\Delta}{t}\right)^{2\phi-2}-\left(\frac{\Delta}{t}\right)^{\phi-1}+\frac{1}{2}\right], \qquad \phi \neq 1, \phi \neq \frac{1}{2}
 \end{split}
 \\
 \begin{split}
  \left\langle r_t^2 \right\rangle=& 2dD_{rw}\left[ t + 2C_0^2\Delta \ln\left( \frac{\Delta}{t}\right)\right]\\ &  +2v^2C_0^2\left[ 2t\Delta+2\Delta^2-4\Delta\left(\Delta t\right)^{\frac{1}{2}} \right], \qquad \phi=\frac{1}{2}
 \end{split}
 \\
 \begin{split}
  \left\langle r_t^2 \right\rangle=& 2dD_{rw}\left[ t + 2C_0^2 \left(\frac{\Delta^2}{t} -\Delta \right)\right]   \\ & + \left[vC_0\Delta \ln\left(\frac{\Delta}{t}\right)\right]^2, \qquad \qquad \qquad \qquad \phi=1.
 \end{split}
\end{align}
\label{msdpowintegrated}
\end{subequations}
These expressions for the MSD are valid when $\Delta \rightarrow
0$. For long crossover times when reorientation probabilities are given
by Eq.~(\ref{pieceprobs}), the MSD is (see Sec.~E in the
Supplementary Information):
\begin{equation}
\left\langle r_t^2 \right\rangle = \begin{cases}(vt)^2 & t\leq t^{\star}, \\
\!\begin{aligned}
& 2dD_{rw}\left[(t-t^{\star})-2\int_{t^{\star}}^tg^2(\tau)\mathrm{d}\tau \right] \\ &
+v^2\left[2\int_{t^{\star}}^t\int_{\tau}^tg(\tau)g(k)\mathrm{d}k\mathrm{d}\tau + t^{\star 2}+ 2t^{\star}\int_{t^{\star}}^tg(\tau)\mathrm{d}\tau\right]
\end{aligned} & t>t^{\star}.
\end{cases}
\label{piecemsd}
\end{equation}

Figure~\ref{theovssims}ii shows a comparison of Eqs.~(\ref{anom_corr}) and
(\ref{msdgenpow}) with LGCA simulations. 
Details on the computational implementation are found in Sec. H of the Supplementary Information.
 From Eqs.~(\ref{msdpowintegrated}) we see that in general $\left\langle
  r_t^2 \right\rangle \sim t \pm t^{1-2\phi}+t^{2\left(1-\phi
  \right)}-t^{1-\phi}$, which defines three regimes:
\begin{enumerate}
\item $\phi<\frac{1}{2}$: superdiffusive regime, arising from the term
  $t^{2\left(1-\phi \right)}$
\item $\frac{1}{2}< \phi < 1$: subdiffusive regime, as the term
  $t^{2\left(1-\phi \right)}$ dominates at short times only
\item $\phi > 1$: normal diffusive regime, as the linear term is the
  dominating term
\end{enumerate}

Figure~\ref{vmmarkpow} shows that, while the VACF decays rapidly in the persistent model,
 in the time-correlated model the VACF decays much more slowly. Additionally, 
the movement is superdiffusive in both models are short times, however this behavior 
is long-lasting in the time-correlated model.

\section*{Generalized time-correlated random walk}

The reorientation probabilities derived for the time-correlated random
walk are only valid for certain time ranges, due to the divergence of
the VACF when $t\rightarrow 0$.  We will now derive a generalized
model which is valid on both short and long time scales. For this
purpose we use what is called the maximum caliber formalism
\cite{calibertheo}, which we introduce briefly.

\subsection*{Rule derivation}

The maximum caliber formalism has been proven successful to derive
models for dynamic systems from data. The
procedure consists in maximizing the entropy over a path of system
evolutions, with the constraint of reproducing certain observables.
The procedure of entropy maximization does not only ensure that the
resulting model contains as few assumptions as possible, but is also
considered the only method of correctly obtaining unknown probability
distributions from known data \cite{calibertheo}. The procedure is as
follows:

Let the path entropy, or caliber, be defined as $
\mathcal{C}=-\sum_{\Gamma}P_{\Gamma}\ln P_{\Gamma}, $ where $\Gamma$
is a possible path followed by the system during its evolution. The
probability of following such a path is given by $P_{\Gamma}$. In the
case of a single random walker, the path is the entire history of
particle velocities $\Gamma=i_0i_1i_2\cdots i_k$ up to the last time
step $k$.  Furthermore, we constrain the unknown probabilities by a
normalization constant and the observed VACF (in this case
Eq.~(\ref{vacf_lgca})). Then the problem translates into optimizing
the functional \cite{caliberappl} $
\tilde{\mathcal{C}}\left[P_{\Gamma}\right]=-\sum_{\Gamma}P_{\Gamma}\ln
P_{\Gamma}+\sum_{j=1}^k \beta(j) \left[
  \sum_{\Gamma}P_{\Gamma}\left(\vec{c}_{i_0} \cdot
    \vec{c}_{i_j}\right)-g(j)\right]+\lambda\left(\sum_{\Gamma}P_{\Gamma}-1\right)
$, where $\beta(j)$ and $\lambda$ are Lagrange multipliers to be
determined. The Lagrange multiplier $\beta(j)$ is given by
$\beta(j)=dg(j)$ (see Sec.~F in the Supplementary Information), and $\lambda$ determines
the normalization constant.

Using the expression for $\beta(j)$ we obtain the reorientation probability
\begin{equation}
 P_{i_k,k}=\frac{1}{z}\exp\left[dg(k)\left(\vec{c}_{i_0} \cdot \vec{c}_{i_k}\right)\right],
 \label{pxboltzmann}
\end{equation}
where $z$ is the normalization constant for the reorientation probability.

If one required not only that the VACF was observed but also that the
autocorrelation function would decay similarly independently of the
start and end points, i.e. $\left\langle \vec{c}_{i_m} \cdot
  \vec{c}_{i_n}\right\rangle=\left\langle \vec{c}_{i_j} \cdot
  \vec{c}_{i_l}\right\rangle$, if $n-m=l-j$, the problem would be
similar with the exception that we would now have $\frac{k(k+1)}{2}$
constraints if the trajectory consists of $k$ time steps.
Analogously, the probabilites would then be given by
\begin{equation}
 P_{\Gamma}=\frac{1}{Z}\exp\left[\sum_{j=1}^k\sum_{m=j-1}^{k-1}dg(j-m)\left(\vec{c}_{i_m} \cdot \vec{c}_{i_j}\right)\right].
 \label{totcorrprobs}
\end{equation}

\subsection*{Model analysis and results}

The VACF can be easily calculated by using the properties of the
partition function $Z$. Given a distribution
$P(x)=\frac{1}{Z}\exp\left[-\beta H(x)\right]$, the expected value of
the function $H(x)$ is
\begin{equation}
 \left\langle H \right\rangle = -\frac{\partial}{\partial \beta} \ln Z.
 \label{ergyz}
\end{equation}

\begin{wrapfigure}{r}{0.4\textwidth}
\begin{center}
\includegraphics[width=70mm]{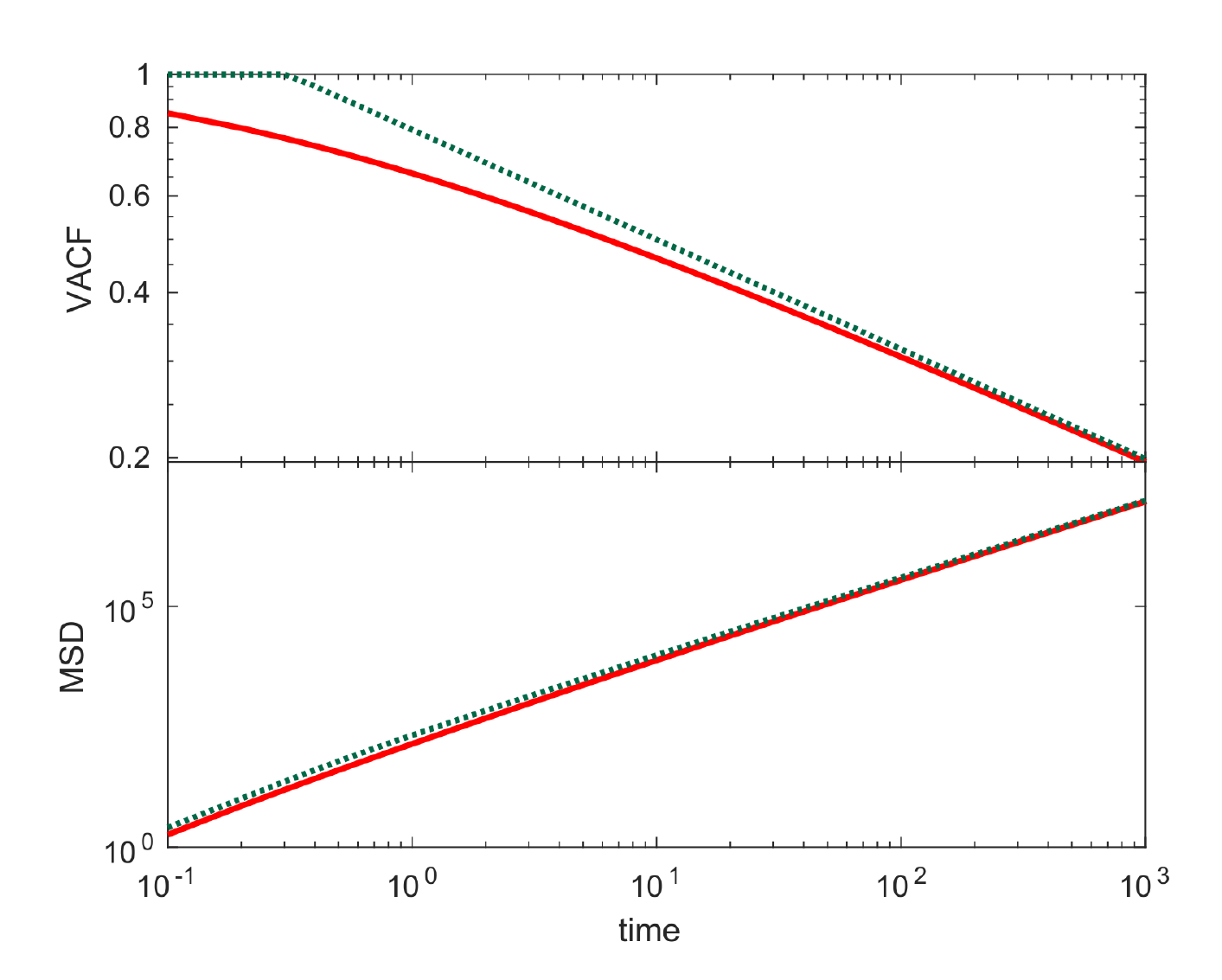}
\end{center}
\caption{Comparison between generalized and time-correlated random
  walk models.  Shown are theoretical results for VACF (top) and MSD
  (bottom) in the time-correlated (dotted green line) and generalized
  time-correlated (solid red line) models. The parameter values are
  $D_{rw}=1$, $v=16$, $C_0=0.5$, $\Delta=10$, $\phi=0.2$.}
\label{vmpowgen}
\end{wrapfigure}

Combining Eqs.~(\ref{lgca-corr}), (\ref{pxboltzmann}) and (\ref{ergyz})
the VACF is given by $g(k) = \frac{\partial}{\partial \beta(k)} \ln
z$, where $\beta(k)=\beta_0(k)$ in the case of of probabilities given
by Eq.~(\ref{totcorrprobs}).  In the case of a
2D square lattice the partition function can be easily calculated as
$z=2\left\{1+\cosh\left[\beta(k)\right]\right\}$.  Taking the
logarithm and differentiating we obtain
$g(k)=\frac{\sinh\left[\beta(k)\right]}{1+\cosh\left[\beta(k)\right]}=\tanh\left[\frac{\beta(k)}{2}\right]$.
If we consider power law correlations and take the limit $\tau
\rightarrow 0$ we obtain 
\begin{equation}
 g(t)=\tanh\left[ C_0\left( \frac{\Delta}{t} \right)^{\phi} \right].
 \label{decaytanpow}
\end{equation}
Eqs.~(\ref{anom_corr}) and (\ref{decaytanpow}), corresponding to the 
time-correlated random walk and generalized time-correlated random walk,
respectively, are visually compared
in Fig.~\ref{vmpowgen}.  A Taylor expansion of $g(t)$ around
$\left(\frac{\Delta}{t}\right)^{\phi}=0$ (i.e.\ valid for
$\Delta\rightarrow 0$ or $k\rightarrow \infty$) shows that up to
second order terms $g(t)\approx C_0\left( \frac{\Delta}{t}
\right)^{\phi}$, as expected. For
$\left(\frac{\Delta}{t}\right)^{\phi} \gg 0$ the VACF decays as a
power law as well (see Sec.~G in the Supplementary Information), $g(t) \approx
G(t)=\tanh(C_0)\left(\frac{\Delta}{t}\right)^{2C_0\phi\csch\left(2C_0\right)}$,
and the difference between both functions behaves as $g(t)-G(t)\propto
\left(\frac{\phi}{\Delta}\right)^2$.
When the probabilites are given by Eq.~(\ref{pxboltzmann}) particle
orientations are independent of one another, as it was the case in the
time-correlated model, and the MSD is given by
Eq.~(\ref{msdind}). Therefore, in the limit $\tau \rightarrow 0$ the MSD
is
\begin{equation}
\begin{split}
 \left\langle r_t^2 \right\rangle & =2dD_{rw}\left\{ t- 2\int_0^t \tanh^2 \left[ \frac{\beta(\tau)}{2} \right] \mathrm{d}\tau \right\} \\ &+ 2v^2 \int_0^t \int_{\tau}^t \tanh \left[ \frac{\beta(\tau)}{2} \right] \tanh \left[ \frac{\beta (k)}{2} \right] \mathrm{d}k \mathrm{d}\tau .
 \end{split}
 \label{contmsdexp}
\end{equation}
If we consider power law correlations and expand in Taylor series
around $\Delta^{\phi}=0$ up to second order, we recover
Eq.~(\ref{msdgenpow}) when $g(t)$ is given by
Eq. \ref{anom_corr}. Therefore the MSD in this model follows similar
regimes as those of the time-correlated model
Eq.~(\ref{msdpowintegrated}). Eqs.~(\ref{piecemsd}) and
(\ref{contmsdexp}) are visually compared in Fig.~\ref{vmpowgen}, and
see Fig.~\ref{theovssims}iii for a comparison of Eqs.~(\ref{decaytanpow}) and
(\ref{contmsdexp}) with LGCA simulations.
Details on the computational implementation are found in Sec. H of the Supplementary Information.

When the probabilities are given by Eq.~(\ref{totcorrprobs}) instead
the MSD is given by Eq.~(\ref{corrmsd}), as all orientation pairs $i$
and $j$ are correlated, where the correlation is given by $g(i-j)$
thus not depending specifically on the values of $i$ and $j$ but on
their difference only. Therefore, for a 2D square lattice and a
power-law decaying VACF in the limit $\tau \rightarrow 0$ the MSD is
given by
 \begin{equation}
 \left\langle r_t^2 \right\rangle =2dD_{rw}t+2v^2\int_{0}^{t}\left(t-\tau \right)\tanh\left[C_0\left(\frac{\Delta}{\tau}\right)^{\phi}\right]\mathrm{d}\tau.
\end{equation}
Again we see that this expression agrees with the formal solution of
the MSD for an overdamped Langevin equation with colored noise
\cite{ChKl12}. In this case, however, the noise correlation is not
decaying exponentially.  If we expand the hyperbolic tangent on the
right hand side around $\left(\frac{\Delta}{t}\right)^{\phi}=0$
(i.e.\ $\Delta \rightarrow 0$ or $t\rightarrow \infty$) up to the second
term and integrate we obtain the MSD as
 \begin{subequations}
\begin{align}
 \begin{split}
  \left\langle r_t^2 \right\rangle=& 2dD_{rw}t +\frac{2C_0}{1-\phi}\left(v\Delta\right)^2\left\{\frac{1}{2-\phi}\left[\left(\frac{\Delta}{t}\right)^{\phi-2}-1\right]+1-\frac{t}{\Delta}\right\}, \qquad \phi \neq 1, \phi \neq 2
 \end{split}
 \\
 \begin{split}
  \left\langle r_t^2 \right\rangle=& 2dD_{rw}t + 2C_0\left(v\Delta\right)^2\left\{\frac{t}{\Delta}\left[\ln\left(\frac{t}{\Delta}\right)-1\right]+1\right\}, \qquad \phi=1
 \end{split}
 \\
 \begin{split}
  \left\langle r_t^2 \right\rangle=& 2dD_{rw}t + 2C_0\left(v\Delta\right)^2\left[ \frac{t}{\Delta}-\ln\left(\frac{t}{\Delta}\right)-1 \right], \qquad \qquad \qquad \qquad \phi=2.
 \end{split}
\end{align}
\label{msdtotcorrintegrated}
\end{subequations}
We conclude from Eqs.~(\ref{msdtotcorrintegrated}) that
$$
\left\langle r_t^2 \right\rangle \sim t \pm t^{2-\phi}.
$$
When $\phi<1$ we have $2-\phi>1$ and the process is
superdiffusive. When $1<\phi<2$ we have $0<2-\phi<1$, at short times
the term $t^{2-\phi}$ dominates, and the process is subdiffusive while
at long times the linear term dominates yielding normal diffusion.
Finally, when $\phi>2$ we have $2-\phi<0$ and the process is
completely normal diffusive.

\section*{Summary and discussion}

The goal of our study was to design a simple model for a single particle moving
with memory in abscence of any environmental cue.
We chose a cellular automaton, specifically, an LGCA because it is a flexible and computationally efficient framework and has the potential to analyze collective behavior in populations of moving particles or cells.
After having introduced an LGCA model for unbiased random walk, we have
derived three different novel time-correlated LGCA models for single
particle migration. 

The subsequent persistent random walk LGCA model was derived from a biophysically\hyp{}motivated
Langevin equation for particle reorientation in an overdamping
environment. We showed that in this model the VACF decays
exponentially. Furthermore, we proved that a particle in this model
moves superdiffusively only at short times while it diffuses normally
in the long time limit. This behavior as well as the expression we
found for the MSD agree with that found by Othmer~\cite{expprevious}
using the telegrapher's equation and also to the one by Chechkin et al. \cite{ChKl12}
for exponential noise.

The time-correlated random walk model was derived by assuming that the specific form of the
VACF is known. We also assumed that reorientation probabilities were completely
independent, and that particles have no preference in turning left or
right. We considered the specific case of a power law\hyp{}decaying
VACF and showed that the MSD exhibits two transitions
when $g(t)\propto t^{-\frac{1}{2}}$ and $g(t)\propto t^{-1}$. For
small exponents the particle moves superdiffusively on every time
scale. At intermediate exponents there are non\hyp{}linear
contributions dominating at short times. For large exponents, all
non\hyp{}linear contributions vanish in the long time limit resulting
in normal diffusion.

Finally, we derived a generalized LGCA model by maximizing the diffusing particle's
path entropy while retaining the constraint of reproducing a certain
VACF. In this model the reorientation probability Eq.~(\ref{pxboltzmann})
is similar to the reorientation probability of the persistent random walk
Eq.~(\ref{lgcatrans}), with some differences: in
Eq.~(\ref{pxboltzmann}) the particle's orientation is compared to its
initial orientation while in Eq.~(\ref{lgcatrans}) the particle's
orientation is compared to the particle's orientation at the previous
time step.  Furthermore, in Eq.~(\ref{lgcatrans}) we have a constant
parameter $\beta$ while in Eq.~(\ref{pxboltzmann}) this parameter decays with time, i.e.\ $\beta \propto g(t)$. We recall that
Eq.~(\ref{lgcatrans}) results from considering a Langevin equation for
the particle's reorientation.  The parameter $\beta$ depends on the
magnitude of the reorienting force, the relaxation constant $\gamma$
(related to friction) and the angular diffusion constant $D_{\theta}$.
A time\hyp{}dependent parameter $\beta(k)$ would be obtained when
considering a generalized Langevin equation resulting from either a
time-dependent reorientation force or friction if these values
changed much more slowly than the time needed for the particle to be
displaced.  Taking this into account, Eq.~(\ref{lgcatrans}) describes
the movement of a particle when reorientations can be performed almost
instantaneously compared to the time required for the particle to move
in space.  Eq.~(\ref{pxboltzmann}) on the other hand describes
movement when the particle keeps moving but needs considerably more
time to change its initial orientation.  When in addition the VACF is
required to be invariant under time translations we showed that the
corresponding MSD time regimes match to those found by Chechkin et al.~\cite{ChKl12} for power law-correlated noise.  

We have verified our analytical results of all constructed models by
comparing them to LGCA computer simulations. In order to derive the analytical
VACF and MSD expressions, we have considered the limit $\tau \rightarrow 0$. 
In this limit, the macroscopic time $t$ remains small even after several time steps.
It stands to reason that, for $t\rightarrow \infty$, $\tau \ll t$, the difference
between our analytical expressions and simulations becomes negligible.

In their present form our LGCA models assume (i) the particle has
constant instantaneous speed $v$; (ii) the particle moves to a
neighboring site at every time step; and (iii) the particle moves on a
regular lattice, which impacts the specific expression of the VACF.
All these models could be extended by considering different instantaneous
speeds, as well as waiting times between subsequent displacements, by using multispeed LGCA
and adding rest (zero velocity) channels, respectively.  Effects of the
lattice regularity on the single particle movement can be compensated
by choosing the sensitivity $\beta$ appropriately in the persistent
random walk model as well as the crossover time $\Delta$ in the
generalized time-correlated model.  In the time-correlated model the
VACF does not depend on the lattice geometry.

Our new models could also be extended to account for external forces acting on the particle, independent from its intrinsic anomalous movement. 
For extending the models, we can consider that particle reorientations are caused by internal correlations of individual cells and by particle interactions.
The probability $P_{n^1,\ldots,n^N,k}^{corr}$ of having $N$ particles in a node with a certain orientation due to internal particle orientations has already been introduced in Eq.~(\ref{confprobs}), while individual particle reorientation probabilities $P_{n^{\ell},k}$ would be given according to one of the models introduced in this work. On the other hand, the probability $P_{n^1,\ldots,n^N,k}^{int}$ of having $N$ particles in a node with a certain orientation due to particle interactions would be a function of other particles' positions and orientations (see~\cite{langlgca}, for examples of such probabilities). If we assume that both probabilities are independent, then the reorientation probability for all particles would simply be $P_{n^1,\ldots,n^N,k}^{tot}=P_{n^1,\ldots,n^N,k}^{corr}\cdot P_{n^1,\ldots,n^N,k}^{int}$. 
 Such an extension could be useful for studying physical systems such as plasma gases~\cite{Bale05}. Furthermore, it would be interesting to construct LGCA models generating L\'evy walks exhibiting non-Gaussian probability density functions~\cite{KlSo11,ZDK15}. More importantly, and tracing back to the original motivation of our work, due to the computational efficiency of LGCAs, our schemes could be applied to model large groups of interacting cells to study the impact of persistence and time correlations in single-cell dynamics on collective phenomena. Highly promising examples are coordination and swarming in bacteria \cite{ARBPHB15}, pluripotent cells during early development \cite{BBG14}, and the emergence of phase transitions in collective cell migration~\cite{MeVi14,McKPL10,Malin17,BYMM16}. Moreover, the non-cellular microenvironment is crucial for cell migration phenomena. Recently, the impact of complex environments on cell dissemination has been studied with a cellular automaton model~\cite{talk17}. It would be interesting to extend the models introduced here to analyze the impact of anomalous dynamics and complex microenvironments on cell dissemination and cancer invasion. \\
The LGCA modeling framework followed in this work is characterized by simplifying the concept of a moving particle to movement in discrete time steps between discrete nodes on a regular lattice, possessing only a finite, discrete set of velocities. On one hand, this “discrete approach” is decidedly more simplified and abstract than “continuous approaches” such as continuous time random walks or fractional diffusion equations. On the other hand, as we have shown in the present work, the LGCA offers an advantage not only in computational efficiency and straightforward multiparticle extension, but also in ease of model analysis. This gets rapidly complex in the aforementioned continuous approaches~\cite{MeKl00,CKW04,MeKl04,KRS08,KlSo11,HoFr13,MJJCB14,ZDK15} but remains feasible in the LGCA, even when dealing with systems of interacting particles~\cite{lgcacarsten,lgca,glioma}. \\
In the era of “Big Data”, there is an abundance of biological data. Single or collective cell motility can be measured in vitro or in vivo via various experimental methods such as in vivo two-photon imaging~\cite{caha03} or cell cytometry~\cite{dAl07}, respectively. In this regard, there is a need for “data-driven” modeling frameworks. Our work comes timely to fulfill this scope by proposing the “data-driven” modeling of single particle superdiffusive behavior without prior knowledge of the mechanisms at work. Such an approach is vital for the study of phenomena whose driving mechanisms are currently unknown or challenging to model~\cite{Zienk15,Chris04,sisan12}.

\section*{Acknowledgements}
The authors thank the Centre for Information Services and High Performance Computing (ZIH) at TU Dresden for providing an excellent infrastructure. The authors acknowledge support by the German Research Foundation and the Open Access Publication Funds of the TU Dresden.The authors would like to thank Anja Vo\ss-B\"ohme, Lutz Brusch, Fabian Rost, Osvaldo Chara, Simon Syga, and Oleksandr Ostrenko for their helpful comments and fruitful discussions. Andreas Deutsch is grateful to the Deutsche Krebshilfe for support. Andreas Deutsch is
supported by the German Research Foundation (Deutsche
Forschungsgemeinschaft) within the projects SFB-TR 79 ``Materials for tissue
regeneration within systemically altered bones'' and Research Cluster of
Excellence ``Center for Advancing Electronics Dresden'' (cfaed). Haralampos Hatzikirou  would like to acknowledge the SYSMIFTA ERACoSysMed grant (031L0085B) for the financial support of this work and the German Federal Ministry of Education and Research within the Measures for the Establishment of Systems Medicine, project SYSIMIT (BMBF eMed project SYSIMIT, FKZ: 01ZX1308D). Josu\'e Manik Nava-Sede\~no is supported by the joint scolarship program DAAD-CONACYT-Regierungsstipendien (50017046) by the German Academic Exchange Service and the National Council on Science and Technology of Mexico.

\section*{Author contribution statement}
All authors formulated the mathematical model. J.M.N.S. performed the analysis and simulations. All authors interpreted results. J.M.N.S. wrote the manuscript with contributions from all authors. All authors read and approved the final manuscript.

\section*{Additional information}
The lattice-gas cellular automata code used in this study is available from the corresponding author on reasonable request.\\
\textbf{Competing financial interests} The authors declare no competing financial interests.

\pagebreak

\centerline{\Large \bf Supplementary information}

\setcounter{equation}{0}
\setcounter{figure}{0}
\setcounter{table}{0}
\setcounter{page}{1}
\renewcommand{\figurename}{Supplementary Figure}
\renewcommand{\tablename}{Supplementary Table}
\renewcommand{\theequation}{S\arabic{equation}}
\renewcommand{\thetable}{S\arabic{table}}
\renewcommand{\thefigure}{S\arabic{figure}}

\begin{appendix}

\section{MSD in correlated systems}
Here we derive the general expression for the MSD when there are
temporal correlations. We use the derivation by Sch\"uring \cite{Schu04}.  To start with, we consider the ``MSD'' as
defined by Shalchi \cite{msdderiv}
$$
\left\langle r_i r_j \right\rangle = \int_{t_0}^t \int_{t_0}^t \left\langle v_i(\tau)v_j(\xi) \right\rangle \mathrm{d}\xi \mathrm{d}\tau
$$
which in LGCA notation are written as
\begin{equation}
 \left\langle r_i r_j \right\rangle=\sum_{n=1}^k \sum_{m=1}^k v^2 \tau^2 \left\langle \left[\hat{e}_i\cdot \vec{c}_{i_n}\right] \left[\hat{e}_j\cdot \vec{c}_{i_m}\right] \right\rangle,
\end{equation}
where $\hat{e}_x$ and $\hat{e}_y$ are the two orthonormal unit vectors
in Cartesian coordinates.  In 2D the diagonal elements are given by:
\begin{equation}
\begin{split}
\left\langle r_x^2 \right\rangle & = \sum_{i=1}^k \sum_{j=1}^k v^2 \tau^2 \left\langle \cos\left(\theta_i\right)\cos\left(\theta_j\right) \right\rangle  \\ &
=\sum_{i=1}^k \sum_{j=1}^k \frac{v^2 \tau^2}{2} \left\langle \cos\left( \theta_i-\theta_j \right)+ \cos\left( \theta_i+\theta_j \right) \right\rangle
\end{split}
\end{equation}
and 
\begin{equation}
\begin{split}
\left\langle r_y^2 \right\rangle & = \sum_{i=1}^k \sum_{j=1}^k v^2 \tau^2 \left\langle \sin\left(\theta_i\right)\sin\left(\theta_j\right) \right\rangle  \\ &
=\sum_{i=1}^k \sum_{j=1}^k \frac{v^2 \tau^2}{2} \left\langle \cos\left( \theta_i-\theta_j \right)- \cos\left( \theta_i+\theta_j \right) \right\rangle,
\end{split}
\end{equation}
where $\theta_k=\mathrm{arg}\left[ \vec{c}_{i_k} \right]$.
Adding them up gives the MSD
\begin{equation}
 \left\langle r^2 \right\rangle=\left\langle r_x^2 \right\rangle+\left\langle r_y^2 \right\rangle = \sum_{i=1}^k \sum_{j=1}^k v^2 \tau^2 \left\langle \cos\left( \theta_i-\theta_j \right) \right\rangle,
\end{equation}
which is just a Taylor-Green-Kubo formula \cite{KlKo02,KKCSG06}.  When
$i=j$ we have $\cos\left( \theta_i-\theta_j \right)=1$, so by taking
these terms out of the sum we get
\begin{equation}
 \left\langle r^2 \right\rangle= kv^2\tau^2 + \sum_{i=1}^k \sum_{j=1}^k v^2 \tau^2 \left\langle \cos\left( \theta_i-\theta_j \right) \right\rangle (1-\delta_{ij}).
\end{equation}
On the one hand $v^2\tau^2=\varepsilon^2$ and on the other $\varepsilon^2=2dD\tau$, so using these relations on the first term on the right hand side yields
\begin{equation}
 \left\langle r^2 \right\rangle= 2dDk\tau + \sum_{i=1}^k \sum_{j=1}^k v^2 \tau^2 \left\langle \cos\left( \theta_i-\theta_j \right) \right\rangle (1-\delta_{ij}).
 \label{markmsdeq}
\end{equation}
Because the cosine is an even function $\cos\left( \theta_i-\theta_j
\right)=\cos\left( \theta_j-\theta_i \right)$, which means that we are
adding two identical terms for every $i\neq j$ (this condition is
already satisfied due to the Kronecker delta).  This situation
allows us to rewrite the limits of the interior sum if we take care
of counting each term twice.  Furthermore, using trigonometric
identities it is possible to expand the cosine on the right hand
side,
\begin{equation}
\begin{split}
 \left\langle r^2 \right\rangle & = 2dDk\tau + 2\sum_{i=1}^k \sum_{j=i}^k v^2 \tau^2 \left\langle \cos\left( \theta_i\right)\cos \left(\theta_j \right) \right\rangle (1-\delta_{ij})  \\
 & + 2\sum_{i=1}^k \sum_{j=i}^k v^2 \tau^2 \left\langle \sin\left( \theta_i \right)\sin\left(\theta_j \right) \right\rangle (1-\delta_{ij}).
\end{split}
 \end{equation}
 If the orientations of the particle are completely uncorrelated from one
 another, then it is possible to write the sums on the right as
 \begin{equation}
\begin{split}
 \left\langle r^2 \right\rangle & = 2dDk\tau + 2\sum_{i=1}^k \sum_{j=i}^k v^2 \tau^2 \left\langle \cos\left( \theta_i\right)\right\rangle\left\langle\cos \left(\theta_j \right) \right\rangle (1-\delta_{ij})  \\
 & + 2\sum_{i=1}^k \sum_{j=i}^k v^2 \tau^2 \left\langle \sin\left( \theta_i \right)\right\rangle\left\langle\sin\left(\theta_j \right) \right\rangle (1-\delta_{ij}).
\end{split}
 \end{equation}
 If the reorientation probabilities are even functions of the angle
 $\theta_k$, then we have that $\left\langle \sin\left( \theta_i
   \right)\right\rangle=0$, and so the second sum on the right hand
 side disappears leaving
\begin{equation}
 \left\langle r^2 \right\rangle = 2dDk\tau + 2\sum_{i=1}^k \sum_{j=i}^k v^2 \tau^2 \left\langle \cos\left( \theta_i\right)\right\rangle\left\langle\cos \left(\theta_j \right) \right\rangle (1-\delta_{ij}).
\end{equation}
We can choose our coordinate system such that $\theta_0=0$ without
loss of generality. With this choice of the coordinate system, we can
rewrite the MSD as
\begin{equation}
 \left\langle r^2 \right\rangle = 2dDk\tau + 2\sum_{i=1}^k \sum_{j=i}^k v^2 \tau^2 \left\langle \cos\left( \theta_i-\theta_0\right)\right\rangle\left\langle\cos \left(\theta_j -\theta_0 \right) \right\rangle (1-\delta_{ij}).
\end{equation}
 Using the definition of the VACF this becomes
 \begin{equation}
 \left\langle r^2 \right\rangle = 2dDk\tau + 2\sum_{i=1}^k \sum_{j=i}^k v^2 \tau^2 g(i)g(j) (1-\delta_{ij}).
\end{equation}
Expanding the difference on the right hand side we get
 \begin{equation}
 \left\langle r^2 \right\rangle = 2dDk\tau + 2\sum_{i=1}^k \sum_{j=i}^k v^2 \tau^2 g(i)g(j)-2\sum_{i=1}^k \sum_{j=i}^k v^2 \tau^2 g(i)g(j)\delta_{ij}.
\end{equation}
The last sum on the right hand side can be simplified, so we get
 \begin{equation}
 \left\langle r^2 \right\rangle = 2dDk\tau + 2\sum_{i=1}^k \sum_{j=i}^k v^2 \tau^2 g(i)g(j)-2\sum_{i=1}^k v^2 \tau^2 g^2(i).
\end{equation}
Finally, using the relation between the intantaneous particle velocity
$v$ and the diffusion constant in the random walk limit $D_{rw}$, and
reordering terms, we obtain
 \begin{equation}
 \left\langle r^2 \right\rangle = 2dD_{rw}\left[k\tau-2\sum_{i=1}^k g^2(i) \tau \right] + 2 v^2 \sum_{i=1}^k \sum_{j=i}^k g(i)g(j)\tau^2.
 \label{memorymsd}
\end{equation}
 
\section{VACF and MSD derivation in the persistent random walk}

First we analytically derive the expected form of the VACF for a
single particle in an LGCA where the lattice is a 2D square lattice.
\subsection*{VACF}
\label{vacfmark}
As mentioned before, the orientation probability is given by
Eq.~(\ref{lgcatrans}).  The VACF is formally defined as
$g(t)=\left\langle \vec{v_0} \cdot \vec{v_t} \right\rangle$, where
$v_0$ and $v_t$ are the velocities of the particle at time 0 and $t$,
respectively.  Using this definition, we can calculate the VACF of a
stochastically moving particle as $$g(t)=\int P(\vec{v},t)\left(
  \vec{v_0} \cdot \vec{v}\right) \mathrm{d}\vec{v},$$ where
$P(\vec{v},t)$ is the probability of the particle having a velocity
$\vec{v}$ at time $t$.

In an LGCA particle velocities are given by the velocity channels they are
located in, which belong to a finite set of unit vectors depending on
the lattice dimension and geometry.  Furthermore, time is also
discrete with time steps of length $\tau$ such that at time step $k$
time has elapsed by $k\tau$. We can then rewrite the definition of the
velocity autocorrelation in the case of an LGCA in the following way
\cite{KlKo02}:
\begin{equation}
 g(t)=\left\langle \vec{v_0} \cdot \vec{v_t}\right\rangle=\sum_{i=1}^{b}P_{i_k,k}\left[ \vec{c}_{i_0} \cdot \vec{c}_{i_k}\right],
\end{equation}
where $\vec{c}_{i_0}$ is the orientation of the particle at time step
$k=0$ and $\vec{c}_{i_k}$ is the orientation of the particle at time step
$k$.  To calculate the VACF, we start by defining a function as:
\begin{equation}
 H=-\vec{c}_{i_{k-1}} \cdot\vec{c}_{i_k},
\end{equation}
where $\vec{c}_{i_k}$ is the particle orientation at time step $k$. Having defined this function, we can rewrite Eq.~(\ref{lgcatrans}) as follows:
\begin{equation}
 P_{i_k,k}=\frac{e^{-\beta H\left( \vec{c}_{i_k},k \right)}}{Z},
\end{equation}
where the partition function is defined as
\begin{equation}
 Z=\sum_{i_k}e^{-\beta H\left( \vec{c}_{i_k},k \right)}.
\end{equation}
The expected value of the function is given by
\begin{equation}
 \left\langle H\right\rangle=\left\langle-\vec{c}_{i_{k-1}} \cdot\vec{c}_{i_k}\right\rangle,
\end{equation}
that is, the energy of the system is the single\hyp{}step correlation.
Due to the distribution of the reorientation probabilites the total
energy can be calculated by the well\hyp{}known relation
\begin{equation}
 \left\langle H\right\rangle=-\frac{\partial}{\partial \beta} \ln Z.
\end{equation}
Using the last two equations, we get an expression for the single step
correlation:
\begin{equation}
 \left\langle\vec{c}_{i_{k-1}}\cdot \vec{c}_{i_k}\right\rangle=\frac{\partial}{\partial \beta} \ln Z.
 \label{emark}
\end{equation}
In this single\hyp{}particle model, the partition function can be easily
calculated. For a 2D square lattice the partition function reads
\begin{equation}
 Z=2\left[1+\cosh(\beta)\right].
 \label{zm2d}
\end{equation}
Substituting Supplementary Eq.~(\ref{zm2d}) into Supplementary Eq.~(\ref{emark}), we have
\begin{equation}
\begin{split}
\left\langle\vec{c}_{i_{k-1}}\cdot \vec{c}_{i_k}\right\rangle & =\frac{\partial}{\partial \beta}\left\{ \ln(2)+\ln\left[1+\cosh(\beta)\right] \right\} \\ & = \frac{\sinh(\beta)}{1+\cosh(\beta)}=\tanh\left(\frac{\beta}{2}\right).
\end{split}
\label{ssmtanh}
\end{equation}
The particle orientations $\vec{c}_{i}$ are normalized vectors. This
allows us to rewrite the single step correlation as
\begin{equation}
\left\langle\vec{c}_{i_{k-1}}\cdot \vec{c}_{i_k}\right\rangle=\left\langle\cos\left( \theta_k - \theta_{k-1} \right)\right\rangle,
\label{ssang}
\end{equation}
and the VACF as
\begin{equation}
 g(k)=\left\langle \cos\left( \theta_k - \theta_0 \right)\right\rangle,
 \label{vacfang}
\end{equation}
where $\theta_k=\mathrm{arg}\left[
  \vec{c}_k\right]$. Supplemenary Eq.~(\ref{vacfang}) can be rewritten by adding
zeroes in the following way:
$$
 g(k)=\left\langle\cos\left( \theta_k - \theta_0 \right)\right\rangle=\left\langle \cos\left[ \theta_k - \theta_0 + \sum_{i=1}^{k-1}\left( \theta_i-\theta_i \right)\right]\right\rangle,
$$
which after rearranging terms, has the form
\begin{equation}
 g(k)=\left\langle \cos\left(\sum_{i=1}^k \theta_i-\theta_{i-1} \right)  \right\rangle.
\end{equation}
Using trigonometric identities and the linearity of the expected value
operator we can expand this expression to
\begin{equation}
\begin{split}
 g(k) & =\left\langle \prod_{i=1}^k \cos\left( \theta_i-\theta_{i-1}\right) \right\rangle \\ & +f \left\{ \left\langle \cos \left( \theta_n-\theta_{n-1} \right) \sin \left( \theta_m-\theta_{m-1} \right) \right\rangle \right\},
 \end{split}
\end{equation}
where $f$ is a sum of expected values of products of sines and
cosines. Because the model is Markovian, the $i$-th orientation is
only correlated with the $(i-1)$-th orientation. This allows us to
write
\begin{equation}
\begin{split}
 g(k) & = \prod_{i=1}^k \left\langle \cos\left( \theta_i-\theta_{i-1}\right)\right\rangle  \\ & +f \left\{ \left\langle \cos \left( \theta_n-\theta_{n-1} \right) \right\rangle \left\langle \sin \left( \theta_m-\theta_{m-1} \right)\right\rangle \right\}.
 \end{split}
\end{equation}
Now, because the reorientation probabilities Eq.~(\ref{lgcatrans}) are
even functions with respect to $\theta_k=\mathrm{arg}\left[ \vec{c}(k)
\right]$, the expected values become
$$
\left\langle \sin \left( \theta_m-\theta_{m-1} \right)\right\rangle=0,
$$
which in turn implies
$$
f \left\{ \left\langle \cos \left( \theta_n-\theta_{n-1} \right) \right\rangle \left\langle \sin \left( \theta_m-\theta_{m-1} \right)\right\rangle \right\}=0.
$$
Using these relations we find that the VACF is given by
\begin{equation}
 g(k) = \prod_{i=1}^k \left\langle \cos\left( \theta_i-\theta_{i-1}\right)\right\rangle.
 \label{almostmvacf}
\end{equation}
Using Supplementary Eqs.~(\ref{ssmtanh}) and (\ref{ssang}) in Supplementary Eq.~(\ref{almostmvacf}) yields
\begin{equation}
 g(k)=\left[ \tanh \left( \frac{\beta}{2}\right) \right]^k,
\end{equation}
which can be written as
\begin{equation}
 g(k)=e^{\alpha k},
\end{equation}
if we define the exponent $\alpha$ as
\begin{equation}
 \alpha=\ln \left[ \tanh \left( \frac{\beta}{2} \right) \right].
 \label{gammsqmark}
\end{equation}
The exponent $\alpha$ depends on the lattice dimension and geometry, as follows:
\begin{itemize}
\item In 1D the exponent is given by:
\begin{equation}
 \alpha=\ln\left[ \tanh \left( \beta \right)\right].
\end{equation}
\item In 2D with a triangular lattice the exponent is:
\begin{equation}
 \alpha = \ln \left( \frac{e^{\frac{\beta}{2}}-e^{-\beta}}{e^{-\beta}+2e^{\frac{\beta}{2}}} \right).
\end{equation}
\item In 2D with a square lattice the exponent is given by:
\begin{equation}
 \alpha=\ln \left[ \tanh \left( \frac{\beta}{2} \right)\right].
\end{equation}
\item In 2D with an hexagonal lattice it takes the form:
\begin{equation}
 \alpha=\ln\left[ \frac{2\sinh\left(\frac{3\beta}{4}\right)\cosh\left(\frac{\beta}{4}\right)}{\cosh(\beta)+2\cosh\left(\frac{\beta}{2}\right)} \right].
\end{equation}
\item With a cubic 3D lattice the exponent reads:
\begin{equation}
 \alpha=\ln \left[\frac{\sinh(\beta)}{\cosh(\beta)+2}\right].
\end{equation}
\end{itemize}

\subsection*{Mean square displacement}
To calculate the MSD of particles performing persistent random walks,
we start with Supplementary Eqs.~(\ref{markmsdeq}) and (\ref{ssang}) and rewrite the
sum limits taking into account that the cosine is an even function to
obtain
\begin{equation}
 \left\langle r^2 \right\rangle = 2dDk\tau + 2v^2\sum_{n=1}^k\sum_{m=n}^k \left\langle \vec{c}_{i_n} \cdot \vec{c}_{i_m} \right\rangle \tau^2 (1-\delta_{nm}).
\end{equation}
The expected value on the right hand side is the $(m-n)$-step correlation. From the previous VACF calculation we know that in the Markovian model
\begin{equation}
 \left\langle \vec{c}_{i_n} \cdot \vec{c}_{i_m} \right\rangle \coloneqq g(m-n) =\left[ \tanh \left( \frac{\beta}{2} \right) \right]^{\mid m-n \mid},
\end{equation}
which can be substituted in the expression of the MSD to obtain
\begin{equation}
 \left\langle r^2 \right\rangle = 2dDk\tau + 2v^2\sum_{n=1}^k\sum_{m=n}^k g(m-n) \tau^2 (1-\delta_{nm}).
\end{equation}
Because the sums on the right hand side only depend on the interval
length $\mid n-m \mid$ and not on the specific values of the indices
$n$ and $m$ we can replace both sums by a sum over all posible
interval lengths.  There are $k-j$ ways to divide an interval of $k$
time steps (because the sums start from $n=1$) into intervals of size
$j$. Taking all into account, the MSD becomes
\begin{equation}
 \left\langle r^2 \right\rangle = 2dDk\tau + 2v^2\sum_{j=1}^k (k-j) g(j) \tau^2,
\end{equation}
which can also be written as
$$
\left\langle r^2 \right\rangle=2dDk\tau+2v^2\sum_{j=1}^{k}\left(k-j\right)e^{j\alpha}\tau^2,
$$
where $\alpha$ is given by Supplementary Eq.~(\ref{gammsqmark}).  We now distribute
the two multiplying time steps $\tau$ on the second term on the right
hand side, and multiply by one the exponent of the exponential
function thus leaving it unchanged:
$$
\left\langle r^2 \right\rangle=2dDk\tau+2v^2\sum_{j=1}^{k}\left(k\tau-j\tau\right)e^{\frac{\alpha}{\tau}j\tau}\tau.
$$
We now use the definitions of the diffusion coefficient and the particle speed to obtain the following expression for the time step length:
$$
\tau=\frac{2dD}{v^2},
$$
and use it to substitute for $\tau$ on the denominator of the exponent
$$
\left\langle r^2 \right\rangle=2dDk\tau+2v^2\sum_{j=1}^{k}\left(k\tau-j\tau\right)e^{\frac{v^2\alpha}{2dD}j\tau}\tau.
$$

\section{VACF in homogeneous Markovian models}
We will now consider a general Markovian model for a single moving particle. The model is then a Markov chain of particle orientations, i.e.\ the particle can transition between different orientations at each time step. 
\begin{definit}
 The state space of the Markov chain is $\mathcal{E}=\left\{\vec{c}_0,\vec{c}_{\pm1}, \cdots, \vec{c}_{\pm n}, \cdots , \vec{c}_{N}\right\}$, where the $2N$ different states are given by
 \begin{equation*}
  \begin{split}
   & \vec{c}_n=\left(\cos\left(\frac{\pi}{n}\right),\sin\left(\frac{\pi}{n}\right)\right), \ n=1,\cdots,N-1, \\
   & \vec{c}_0=(1,0), \\
   & \vec{c}_N=(-1,0).
  \end{split}
 \end{equation*}
\end{definit}
\begin{definit}
 The state space subset $\mathcal{E}_0$ is defined as
 \begin{equation*}
  \mathcal{E}_0\coloneqq \left\{\vec{c}_{\pm1},\cdots,\vec{c}_{\pm(N-1)}\right\}.
 \end{equation*}
\end{definit}
If the space is isotropic, then it is reasonable to require that the
probability of the particle turning left or right be
identical. Furthermore, we assume that the probability of turning does
not depend on the specific time step, i.e.\ that the Markov process is
homogeneous.
\begin{definit}
 The Markov chain is the stochastic process $\left\{X(k):k\in \mathbb{N}\right\}$ where the reorientation probabilities are given by
 \begin{equation*}
  P\left(X(k+1)=\vec{c}_m\mid X(k)=\vec{c}_n\right)\coloneqq P\left( \theta \coloneqq\mathrm{arg}(\vec{c}_n,\vec{c}_m)\right), 
 \end{equation*}
where $P(\theta)=P(-\theta)$ for $0<\mid\theta\mid<\pi$, and the initial condition $X(0)=\vec{c}_0$.
\end{definit}
We will use the following shorthand notation: $P(0)\coloneqq p_0$, $P(\frac{\pi}{n})=P(-\frac{\pi}{n})\coloneqq p_n$, and $P(\pi) \coloneqq p_N$.
\begin{definit}
 The rotation matrix $\mathbf{A_n}$ is given by
 \begin{equation*}
  \mathbf{A_n}=\begin{pmatrix}
                \cos\left(\frac{\pi}{n}\right) & -\sin\left(\frac{\pi}{n}\right) \\
                \sin\left(\frac{\pi}{n}\right) & \cos\left(\frac{\pi}{n}\right)
               \end{pmatrix},
 \end{equation*}
 such that
 \begin{equation*}
  \mathbf{A_n}\left(\!\begin{array}{c}
               1 \\
               0
              \end{array}
              \!\right)=\left(\!\begin{array}{c}
               \cos\left(\frac{\pi}{n}\right) \\
               \sin\left(\frac{\pi}{n}\right)
              \end{array}
              \!\right)=\vec{c}_n
 \end{equation*}
and 
\begin{equation*}
 \mathbf{A_n}\left(\!\begin{array}{c}
               \cos(\phi) \\
               \sin(\phi)
              \end{array}
              \!\right)=\left(\!\begin{array}{c}
               \cos\left(\frac{\pi}{n}\right)\cos(\phi)-\sin\left(\frac{\pi}{n}\right)\sin(\phi) \\
               \sin\left(\frac{\pi}{n}\right)\cos(\phi)+\cos\left(\frac{\pi}{n}\right)\sin(\phi)
              \end{array}
              \!\right)=\left(\!\begin{array}{c}
               \cos(\frac{\pi}{n}+\phi) \\
               \sin(\frac{\pi}{n}+\phi)
              \end{array}
              \!\right).
\end{equation*}
\end{definit}

\begin{definit}
 The velocity autocorrelation function (VACF) is given by
 \begin{equation*}
  g_k=\left\langle X(0)\cdot X(k)\right\rangle=\sum_{\vec{v}\in \mathcal{E}}\left(\vec{c}_0 \cdot\vec{v}\right)P^k\left(\vec{v}\right)
 \end{equation*}
\end{definit}
\begin{theo}
 The velocity autocorrelation function of a particle whose orientations are given by a homogeneous, symmetric Markov chain is either delta-correlated, i.e. $g_k=\delta_{0,k}$, where $\delta$ is the Kronecker delta; alternating, i.e. $g_k=(-1)^ka^k$, $a\in\mathbb{R}^+$; or exponentially decaying, i.e. $g_k=e^{\alpha k}$, $\alpha\leq0$.
\end{theo}
 \begin{proof}
  The proof is by induction.
  \begin{equation*}
   \begin{split}
    g_1 = & \sum_{\vec{v}\in\mathcal{E}}(\vec{c}_0\cdot\vec{v})P(\vec{v})=\sum_{\vec{v}=\vec{c}_0,\vec{c}_N}(\vec{c}_0\cdot\vec{v})P(\vec{v})+\sum_{\vec{v}\in\mathcal{E}_0}(\vec{c}_0\cdot\vec{v})P(\vec{v})
    =P(0)-P(\pi)+\sum_{i=1}^{N-1}\left[\left(\vec{c}_0\cdot\vec{c}_i\right)P\left(\frac{\pi}{i}\right)+\left(\vec{c}_0\cdot\vec{c}_{-i}\right)P\left(-\frac{\pi}{i}\right)\right] \\ &
    p_0-p_N+\sum_{i=1}^{N-1} \left[\cos\left(\frac{\pi}{i}\right)p_i+\cos\left(-\frac{\pi}{i}\right)p_i\right]=p_0-p_N+2\sum_{i=1}^{N-1}\cos\left(\frac{\pi}{i}\right)p_i\coloneqq a.
   \end{split}
  \end{equation*}
  Using the Chapman-Kolmogorov equation, we can calculate the VACF at
  the time step $k+1$
  \begin{equation*}
    g_{k+1}=\sum_{\vec{v}\in\mathcal{E}}(\vec{c}_0\cdot\vec{v})P^{k+1}(\vec{v})=\sum_{\vec{v}\in\mathcal{E}}(\vec{c}_0\cdot\vec{v})\sum_{\vec{u}\in\mathcal{E}}P^k(\vec{u})P(\vec{v}\mid\vec{u})
    =\sum_{\vec{u}\in\mathcal{E}}P^k(\vec{u})\sum_{\vec{v}\in\mathcal{E}}(\vec{c}_0\cdot\vec{v})P(\vec{v}\mid\vec{u}).
  \end{equation*}
  We now expand the second sum on the right hand side of the equation
  \begin{equation*}
   \begin{split}
    \sum_{\vec{v}\in\mathcal{E}}(\vec{c}_0\cdot\vec{v})P(\vec{v}\mid\vec{u}) = & 
    (\vec{c}_0\cdot\vec{u})p_0-(\vec{c}_0\cdot\vec{u})p_N+\sum_{i=1}^{N-1}\left[\left(\vec{c}_0\cdot\mathbf{A_i}\vec{u}\right)P\left(\mathbf{A_i}\vec{u}\mid\vec{u}\right)+\left(\vec{c}_0\cdot\mathbf{A_{-i}}\vec{u}\right)P\left(\mathbf{A_{-i}}\vec{u}\mid\vec{u}\right)\right] \\ &
    =(\vec{c}_0\cdot\vec{u})p_0-(\vec{c}_0\cdot\vec{u})p_N+\sum_{i=1}^{N-1}\left[\left(\!\begin{array}{c}
               \cos\left(\frac{\pi}{i}\right) \\
               -\sin\left(\frac{\pi}{i}\right)
              \end{array}
              \!\right)^T\vec{u}p_i+\left(\!\begin{array}{c}
               \cos\left(\frac{\pi}{i}\right) \\
               \sin\left(\frac{\pi}{i}\right)
              \end{array}
              \!\right)^T\vec{u}p_i\right] \\ &
              =(\vec{c}_0\cdot\vec{u})p_0-(\vec{c}_0\cdot\vec{u})p_N+\sum_{i=1}^{N-1}\left(\!\begin{array}{c}
               2\cos\left(\frac{\pi}{i}\right) \\
               0
              \end{array}
              \!\right)^T\vec{u}p_i \\ &
              =(\vec{c}_0\cdot\vec{u})p_0-(\vec{c}_0\cdot\vec{u})p_N+2\sum_{i=1}^{N-1}\cos\left(\frac{\pi}{i}\right)(\vec{c}_0\cdot\vec{u})p_i \\ &
              =(\vec{c}_0\cdot\vec{u})\left[p_0-p_N+2\sum_{i=1}^{N-1}\cos\left(\frac{\pi}{i}\right)p_i\right]=(\vec{c}_0\cdot\vec{u})a.
   \end{split}
  \end{equation*}
Inserting this expression back into the VACF yields
\begin{equation*}
 g_{k+1}=\sum_{\vec{u}\in\mathcal{E}}P^k(\vec{u})\sum_{\vec{v}\in\mathcal{E}}(\vec{c}_0\cdot\vec{v})P(\vec{v}\mid\vec{u})=\sum_{\vec{u}\in\mathcal{E}}P^k(\vec{u})(\vec{c}_0\cdot\vec{u})a=g_ka=a^ka=a^{(k+1)}.
\end{equation*}
We can rewrite $a$ as $a=\sum_{\theta}p(\theta)\cos \theta$, where $\theta=\mathrm{arg}(\vec{c}_0,\vec{c}_n)$, $\forall \vec{c}_n \in \mathcal{E}$. Using the fact that $0\leq p(\theta)\leq1$ and $\sum_{\theta}p(\theta)=1$, we have
\begin{equation*}
 -1\leq\cos(\theta)\leq1\implies -p(\theta)\leq p(\theta)\cos(\theta)\leq p(\theta) \implies -1\leq\sum_\theta p(\theta)\cos \theta\leq1 \therefore -1\leq a\leq 1.
\end{equation*}
 We have three cases:
 \begin{itemize}
  \item $-1\leq a <0$, then $a=-1\mid a\mid$ and $g_k=a^k=(-1)^k\mid a\mid^k$.
  \item $a=0$, then $g_k=a^k=0$, $k\neq0$.
  \item $0<a\leq1$ then $g_k=a^k=e^{k\ln(a)}=e^{\alpha k}$, where $\alpha=\ln(a)$. $0<a\leq1\implies-\infty<\alpha\leq0$.
 \end{itemize}
 \end{proof}

%
%
%

\section{Time correlated random walk: rule derivation for different dimensions and geometries}

\subsection*{One dimension}
We will now sketch our method for obtaining the reorientation
probabilities $P_{i_k,k}$ in 1D.  We start by expanding $g(k)$ for the
first two time steps after $k\tau=t\geq \Delta$ (see
Eq.~(\ref{anom_corr})) for a 1D lattice. We will denote by the
subscript $f$ the lattice direction parallel to the original
orientation of the particle. Similarly, the subscript $r$ denotes the
direction opposite to the original orientation of the particle.  Numerical
subscripts denote the time step at which the reorientation probability is
evaluated.
\paragraph{Time step $k=1$}
Only two trajectories are possible after one time step. Their
probabilities are given by $P_{f,1}$ and $P_{r,1}$. The normalization
condition for these probabilties reads
\begin{equation}
 P_{f,1}+P_{r,1}=1.
 \label{norm1d1k}
\end{equation}
We now expand the VACF:
\begin{equation}
 \sum_{i=1}^{1}P_{i_k,k}\left[ \vec{c}_{i_0} \cdot \vec{c}_{i_k}\right]=P_{f,1}-P_{r,1}=g(1).
 \label{c1d1k}
\end{equation}
We can substitute $P_{r,1}$ from Supplementary Eq.~(\ref{norm1d1k}) into Supplementary Eq.~(\ref{c1d1k})
$$
P_{f,1}-P_{r,1}=P_{f,1}-(1-P_{f,1})=2P_{f,1}-1=g(1).
$$
Rearranging terms we obtain the probability for having the same orientation as originally to
\begin{equation}
 P_{f,1}=\frac{1+g(1)}{2}.
 \label{pf1d1k}
\end{equation}
Substituting Supplementary Eq.~(\ref{pf1d1k}) into Supplementary Eq.~(\ref{norm1d1k}) we obtain
the probability for the particle to turn around after the first time step,
\begin{equation}
 P_{r,1}=\frac{1-g(1)}{2}.
 \label{pb1d1k}
\end{equation}
After inspection of Supplementary Eqs.~(\ref{pf1d1k}) and (\ref{pb1d1k}), and recalling
that in the 1D lattice $c_1=1$ and $c_2=-1$, both probabilities can be
written as a single expression,
\begin{equation}
 P_{i_1,1}=\frac{1+\left[ c_{i_0} \cdot c_{i_1} \right] g(1)}{2},
 \label{px1d1k}
\end{equation}
where $i$ is a placeholder variable for either $f$ or $r$.
\paragraph{Time step $k=2$}
After two time steps, we have four different possible paths for the
particle, with four different probabilities.  If we assume the
probabilities at each time can be written as $P_{ff}=P_{f,1}P_{f,2}$,
we can expand the VACF to obtain:
\begin{equation*}
 P_{f,1}P_{f,2}-P_{f,1}P_{r,2}+P_{r,1}P_{f,2}-P_{r,1}P_{r,2}=\left(P_{f,2}-P_{r,2}\right)\left(P_{f,1}+P_{r,1}\right)=g(2),
\end{equation*}
which by employing Supplementary Eq.~(\ref{norm1d1k}) can simplified to:
\begin{equation}
 P_{f,2}-P_{r,2}=g(2).
 \label{c1d2k}
\end{equation}
Given that the probabilities in the previous time step were
normalized, it is sufficient to require that the probabilities in the
current time step be normalized:
\begin{equation}
 P_{f,2}+P_{r,2}=1.
 \label{norm1d2k}
\end{equation}
Inspecting Supplementary Eqs.~(\ref{c1d2k}) and (\ref{norm1d2k}) and comparing them
with Supplementary Eqs.~(\ref{norm1d1k}) and (\ref{c1d1k}) we can see that they are
identical except for the evaluation of $g(k)$. Therefore, for the
second time step it holds that
\begin{equation}
 P_{i_2,2}=\frac{1+\left[ c_{i_0} \cdot c_{i_2} \right] g(2)}{2}.
 \label{px1d2k}
\end{equation}
\paragraph{Any $k$}
It is easy to see that for further times we can always assume that
the probabilities are uncorrelated so that only the last orientation
in the particle's orientation history is relevant for the calculation.  If
we do, Supplementary Eqs.~(\ref{px1d1k}) and (\ref{px1d2k}) can be generalized for
any time step $k$ in the following way:
\begin{equation}
 P_{i_k,k}=\frac{1+\left[ c_{i_0} \cdot c_{i_k} \right] g(k)}{2}.
 \label{px1d}
\end{equation}
\subsection*{Two dimensions: Triangular lattice}
We will repeat the calculation we did in 1D now in 2D for two
different lattice geometries to identify possible dependencies on the
lattice dimension and/or geometry.
\paragraph{Time step $k=1$}
We have three possible lattice directions with lattice vectors given
by either $\vec{c_1}=\left(1,0\right)$,
$\vec{c_2}=\left(-\frac{1}{2},\frac{\sqrt{3}}{2}\right)$,
$\vec{c_3}=\left(-\frac{1}{2},-\frac{\sqrt{3}}{2}\right)$, or
$\vec{c_1}=\left(\frac{1}{2},\frac{\sqrt{3}}{2}\right)$, $\vec{c_2}=\left(-1,0\right)$,
$\vec{c_3}=\left(\frac{1}{2},-\frac{\sqrt{3}}{2}\right)$ on alternating nodes.
In the first time step there are three possible paths given by
$P_{r,1}$, $P_{a,1}$, and $P_{u,1}$, where $P_{r,1}$ is the
probability to reverse orientation. The normalization condition is, in
this case, given by
\begin{equation}
 P_{r,1}+P_{u,1}+P_{a,1}=1,
 \label{norm2d1ktri}
\end{equation}
while the VACF is given by
\begin{equation}
 \frac{1}{2}(P_{u,1}+P_{a,1})-P_{r,1}=g(1).
 \label{c2d1ktri}
\end{equation}
We need to make an assumption to continue, as there are more variables
than equations. We assume that the probability of turning left or
right is identical,
\begin{equation}
 P_{u,1}=P_{a,1}\coloneqq P_{f,1}.
 \label{ass4}
\end{equation}
Under this assumption we can rewrite Supplementary Eq.~(\ref{norm2d1ktri}) as
\begin{equation}
 P_{r,1}+2P_{f,1}=1,
 \label{gnorm2d1ktri}
\end{equation}
and Supplementary Eq.~(\ref{c2d1ktri}) as
\begin{equation}
 P_{f,1}-P_{r,1}=g(1).
 \label{gc2d1ktri}
\end{equation}
Substituting $P_{f,1}$ from Supplementary Eq.~(\ref{gc2d1ktri}) into
Supplementary Eq.~(\ref{gnorm2d1ktri}) we obtain
$$
P_{r,1}+2(P_{r,1}+g(1))=3P_{r,1}+2g(1)=1,
$$
which, after rearranging, gives the expression for the probability of
the particle to go back:
\begin{equation}
 P_{r,1}=\frac{1-2g(1)}{3}.
 \label{pb2d1ktri}
\end{equation}
Using Supplementary Eq.~(\ref{pb2d1ktri}) in Supplementary Eq.~(\ref{gnorm2d1ktri}) we obtain the
probability
\begin{equation}
 P_{f,1}=\frac{1+g(1)}{3}.
 \label{pf2d1ktri}
\end{equation}
Examining Supplementary Eqs.~(\ref{ass4}), (\ref{pb2d1ktri}) and (\ref{pf2d1ktri})
we can summarize them as
\begin{equation}
 P_{i_1,1}=\frac{1+2\left[ \vec{c}_{i_0} \cdot \vec{c}_{i_1} \right] g(1)}{3}.
 \label{px2d1ktri}
\end{equation}
\paragraph{Time step $k=2$}
In this case there are 9 different possible orientation histories with
9 different probabilities.  If we now denote by $f$ and $r$ the
lattice directions parallel and antiparallel to the original particle
orientation, respectively, and by $u$ and $a$ the remaining lattice
directions and assume that the probabilities are uncorrelated, we
require that probabilites at the present time step are normalized:
\begin{equation}
 P_{u,2}+P_{a,2}+P_{f,2}=1,
  \label{norm2d2ktri}
\end{equation}
while the VACF has the form
\begin{equation*}
\begin{split}
  & P_{u,1}P_{f,2}-\frac{1}{2}P_{u,1}P_{u,2}-\frac{1}{2}P_{u,1}P_{a,2}+P_{a,1}P_{f,2} \\ &
 -\frac{1}{2}P_{a,1}P_{u,2}-\frac{1}{2}P_{a,1}P_{a,2}+P_{r,1}P_{f,2}-\frac{1}{2}P_{r,1}P_{u,2}-\frac{1}{2}P_{r,1}P_{a,2} \\ &
 =\left[ P_{f,2}- \frac{1}{2} \left( P_{u,2}+P_{a,2} \right) \right] \left( P_{u,1}+P_{a,1}+P_{r,1} \right)=g(2),
\end{split}
\end{equation*}
which, by Supplementary Eq.~(\ref{norm2d1ktri}), is simplified to:
\begin{equation}
 P_{f,2}- \frac{1}{2} \left( P_{u,2}+P_{a,2} \right)=g(2).
 \label{c2d2ktri}
\end{equation}
To continue, we impose the isotropy condition Supplementary Eq.~(\ref{ass4})
denoting by $P_{r,2}$ the probabilities $P_{u,2}$ and $P_{a,2}$. With
these assumptions the normalization condition reads
\begin{equation}
 P_{f,2}+2P_{r,2}=1,
 \label{gnorm2d2ktri}
\end{equation}
while the VACF is now
\begin{equation}
 P_{f,2}-P_{r,2}=g(2).
 \label{gc2d2ktri}
\end{equation}
Inserting $P_{f,2}$ from Supplementary Eq.~(\ref{gnorm2d2ktri}) into
Supplementary Eq.~(\ref{gc2d2ktri}) we obtain the probability $P_{r,2}$:
\begin{equation}
 P_{r,2}=\frac{1-g(2)}{3}
 \label{pb2d2ktri}
\end{equation}
and, using the normalization condition Supplementary Eq.~(\ref{gnorm2d2ktri}) we
obtain the probability $P_{f,2}$:
\begin{equation}
 P_{f,2}=\frac{1+2g(2)}{3}.
 \label{pf2d2ktri}
\end{equation}
These probabilities can be written in the general form
\begin{equation}
 P_{i_2,2}=\frac{1+2\left[ \vec{c}_{i_0} \cdot \vec{c}_{i_2} \right] g(2)}{3}.
 \label{px2d2ktri}
\end{equation}
\paragraph{Any $k$}
From Supplementary Eqs.~(\ref{px2d1ktri}) and (\ref{px2d2ktri}) we can see that for any
further time step $k$ and making the same assumptions as before the
probabilities are given by
\begin{equation}
 P_{i_k,k}=\frac{1+2\left[ \vec{c}_{i_0} \cdot \vec{c}_{i_k} \right] g(k)}{3}.
 \label{px2dtri}
\end{equation}
\subsection*{Two dimensions: Square lattice}
\paragraph{Time step $k=1$}
There are four possible lattice directions with lattice vectors
$\vec{c_1}=(1,0)$, $\vec{c_2}=(0,1)$, $\vec{c_3}=(-1,0)$, and
$\vec{c_4}=(0,-1)$. Therefore there are four possible probabilities,
so the normalization condition reads
\begin{equation}
 P_{r,1}+P_{f,1}+P_{u,1}+P_{a,1}=1,
 \label{norm2d1ksq}
\end{equation}
where $P_{u,1}$ and $P_{a,1}$ are the probabilities of going in the
two directions orthogonal to the original orientation of the particle. We
now expand the VACF to obtain
\begin{equation}
 P_{f,1}-P_{r,1}=g(1).
 \label{c2d1ksq}
\end{equation}
Right from the start we have more variables than equations, so we
need to make one more assumption in order to continue with the
derivation. To simplify we assume the following:
\begin{equation}
 P_{u,1}=P_{a,1} \coloneqq \frac{1}{4}.
 \label{ass2}
\end{equation}
With this assumption the normalization condition becomes
\begin{equation}
 P_{f,1}+P_{r,1}=\frac{1}{2}.
 \label{gnorm2d1ksq}
\end{equation}
Inserting Supplementary Eq.~(\ref{gnorm2d1ksq}) into Supplementary Eq.~(\ref{c2d1ksq}) we obtain
$$
P_{f,1}-(\frac{1}{2}-P_{f,1})=2P_{f,1}-\frac{1}{2}=g(1).
$$
Rearranging terms we obtain the probability
\begin{equation}
 P_{f,1}=\frac{1+2g(1)}{4}.
 \label{pf2d1ksq}
\end{equation}
Inserting Supplementary Eq.~(\ref{pf2d1ksq}) into Supplementary Eq.~(\ref{gnorm2d1ksq}) we obtain
the remaining probability
\begin{equation}
 P_{r,1}=\frac{1-2g(1)}{4}.
 \label{pb2d1ksq}
\end{equation}
Supplementary Equations~(\ref{ass2}), (\ref{pf2d1ksq}) and (\ref{pb2d1ksq}) can then be
summarized as
\begin{equation}
 P_{i_1,1}=\frac{1+2\left[ \vec{c}_{i_0} \cdot \vec{c}_{i_1} \right] g(1)}{4}.
 \label{px2d1ksq}
\end{equation}
\paragraph{Time step $k=2$}
There are now 16 different possible histories for the traveling
particle. As before we assume that the probabilities are uncorrelated
which, together with Supplementary Eq.~(\ref{norm2d1ksq}), allows us to write the
normalization condition as
\begin{equation}
\begin{split}
 P_{r,2}+P_{f,2}+P_{u,2}+P_{a,2}=1,
 \end{split}
 \label{norm2d2ksq}
\end{equation}
and the VACF now is
\begin{equation*}
\begin{split}
 & P_{f,1}P_{f,2}+P_{u,1}P_{f,2}+P_{r,1}P_{f,2}+P_{a,1}P_{f,2}-\\
 & (P_{f,1}P_{r,2}+P_{u,1}P_{r,2}+P_{r,1}P_{r,2}+P_{a_1}P_{r,2})= \\ &
 \left( P_{f,2}-P_{r,2} \right)\left( P_{f,1}+P_{u,1}+P_{r,1}+P_{a,1} \right)=g(2),
 \end{split}
\end{equation*}
which, by using Supplementary Eq.~(\ref{norm2d1ksq}), is simplified to
\begin{equation}
P_{f,2}-P_{r,2}=g(2).
  \label{c2d2ksq}
\end{equation}
We see that Supplementary Eqs.~(\ref{norm2d1ksq}) and (\ref{norm2d2ksq}), and
(\ref{c2d1ksq}) and (\ref{c2d2ksq}) are practically identical. Therefore,
by making the same assumptions, we arrive at the following
expression for the probabilities at $k=2$:
\begin{equation}
 P_{i_2,2}=\frac{1+2\left[ \vec{c}_{i_0} \cdot \vec{c}_{i_2} \right] g(2)}{4}.
\end{equation}
\paragraph{Any $k$}
If we continue making the assumptions we have done until now, the probabilities can be generalized in a straightforward way as
\begin{equation}
 P_{i_k,k}=\frac{1+2\left[ \vec{c}_{i_0} \cdot \vec{c}_{i_k} \right] g(k)}{4}.
 \label{px2dsq}
\end{equation}
\subsection*{Three dimensions: cubic lattice}
\paragraph{Time step $k=1$}
In this case we have six lattice directions given by
$\vec{c_1}=(1,0,0)$, $\vec{c_2}=(0,1,0)$, $\vec{c_3}=(0,0,1)$,
$\vec{c_4}=(-1,0,0)$, $\vec{c_5}=(0,-1,0)$, and $\vec{c_6}=(0,0,-1)$.
We denote the lattice direction parallel to the initial orientation
with the subindex $f$, the contrary direction by $r$ and the rest by
$u$, $a$, $d$, and $s$. The normalization condition is
\begin{equation}
 P_{f,1}+P_{u,1}+P_{r,1}+P_{d,1}+P_{a,1}+P_{s,1}=1.
 \label{norm3d1k}
\end{equation}
The VACF is given by
\begin{equation}
 P_{f,1}-P_{r,1}=g(1).
 \label{c3d1k}
\end{equation}
Similarly as in the case of the square lattice, we impose the
following condition which allows deriving the reorientation
probabilities:
\begin{equation}
 P_{u,1}=P_{a,1}=P_{s,1}=P_{d,1}=\frac{1}{6},
 \label{ass5}
\end{equation}
which enables us to simplify the normalization condition in the
following way:
\begin{equation}
 P_{f,1}+P_{r,1}=\frac{1}{3}.
 \label{gnorm3d1k}
\end{equation}
Using Supplementary Eq.~(\ref{gnorm3d1k}) to substitute $P_{r,1}$ into
Supplementary Eq.~(\ref{c3d1k}) we obtain
$$
P_{f,1}-(\frac{1}{3}-P_{f,1})=2P_{f,1}-\frac{1}{3}=g(1),
$$
which, after rearranging terms yields the probability
\begin{equation}
 P_{f_1}=\frac{1+3g(1)}{6}.
 \label{pf3d1k}
\end{equation}
Now, using Supplementary Eq.~(\ref{gnorm3d1k}) we can obtain the remaining probability
\begin{equation}
 P_{r,1}=\frac{1-3g(1)}{6}.
 \label{pb3d1k}
\end{equation}
Examining Supplementary Eqs.~(\ref{ass5}), (\ref{pf3d1k}) and (\ref{pb3d1k}) we
arrive at the general expression
\begin{equation}
 P_{i_1,1}=\frac{1+3\left[ \vec{c}_{i_0} \cdot \vec{c}_{i_1} \right] g(1)}{6}.
 \label{px3d1k}
\end{equation}
\paragraph{Any $k$}
As done before we can continue the process for further times and,
making the same assumptions, we arrive at an equation as
Supplementary Eq.~(\ref{px3d1k}) for any time $k$:
\begin{equation}
 P_{i_k,k}=\frac{1+3\left[ \vec{c}_{i_0} \cdot \vec{c}_{i_k} \right] g(k)}{6}.
 \label{px3d}
\end{equation}
\subsection*{Any dimension, any lattice geometry, any time}
Now that probabilities were derived for several dimensions,
geometries, and times, we can see from Supplementary Eqs.~(\ref{px1d}), (\ref{px2dtri}),
(\ref{px2dsq}) and (\ref{px3d}) that the general form of the
probabilities is given by
\begin{equation}
 P_{i_k,k}=\frac{1+d\left[ \vec{c}_{i_0} \cdot \vec{c}_{i_k} \right] g(k)}{b},
\end{equation}
where $d$ is the spatial dimension and $b$ is the number of nearest neighbors.

\section {MSD of the piecewise process}
We will now calculate the MSD using probabilities such that the VACF is a power-law decaying piecewise function defined as
\begin{equation}
g(t)= \begin{cases} 1 & t\leq t^{\star} \\  C_0\left(\frac{\Delta}{t}\right)^{\phi} & t>t^{\star} \end{cases},
\end{equation}
where $t^{\star}$ is such that $
C_0\left(\frac{\Delta}{t^{\star}}\right)^{\phi}=1$. It is
straightforward to see that the probabilities which define such a VACF
obey
\begin{equation}
P_{i_k,k}= \begin{cases} \delta_{i,i_0} & k \leq \omega \\ \frac{1+d\left[ \vec{c}_{i_0} \cdot \vec{c}_{i_k} \right] g(k)}{b} & k>\omega\end{cases}
\label{probpiece}
\end{equation}
where $i_0$ is the index of the velocity channel the particle started in
and $\omega$ is such that $t^{\star}=\omega \tau$.  It is easy to see
that in the first $\omega$ time steps the MSD is defined by
$$
\left\langle r^2 \right\rangle(k)=k^2\varepsilon^2,
$$
or, using the definition of the particle speed and taking the limit $\tau
\rightarrow 0$:
\begin{equation}
\left\langle r^2 \right\rangle(t)=(vt)^2.
\label{msdminus}
\end{equation}
We will now calculate the MSD for time steps greater than
$\omega$. The calculation will be made for a 1D lattice, but the
results are identical for any dimension and lattice geometry. To ease
notation, we will omit any subindices refering to time steps $k\leq
\omega$, as we know that, given Supplementary Eq.~(\ref{probpiece}), only those
trajectories where the first $\omega$ orientations of the particle are
identical to the original orientation of the particle have non-zero
probabilities.
\paragraph{$\omega+1$}
At the first time step after $\omega$ time steps have elapsed, we find
that the MSD is given by
$$
\left\langle r^2 \right\rangle (\omega+1)=r_f^2 P_{f,\omega+1}+r_r^2 P_{r,\omega+1},
$$
where the displacements are $r_f^2=(\omega+1)^2 \varepsilon^2$ and
$r_r^2=(\omega-1)^2 \varepsilon^2$.  Using Supplementary Eq.~(\ref{probpiece}) and
substituting the square displacements we obtain
\begin{equation*}
\begin{split}
\left\langle r^2 \right\rangle (\omega+1) & = (\omega+1)^2 \varepsilon^2 \left[ \frac{1+g\left( \omega+1 \right)}{2} \right] +  (\omega-1)^2 \varepsilon^2 \left[ \frac{1-g\left( \omega+1 \right)}{2} \right] \\ & 
=\frac{\varepsilon^2}{2}\left\{ (\omega^2 +2\omega +1)[1+g(\omega+1)]+(\omega^2-2\omega+1)[1-g(\omega+1)] \right\} \\ &
=\frac{\varepsilon^2}{2} \left\{ 2\omega^2+2+2\omega[1+g(\omega+1)-1+g(\omega+1) ] \right\}
\end{split}
\end{equation*}
which reduces to
\begin{equation}
\left\langle r^2\right\rangle(\omega+1) =\varepsilon^2[\omega^2+1+2\omega g(\omega+1) ].
\end{equation}
\paragraph{$\omega+2$}
Now, the MSD can be expanded in the following way:
\begin{equation*}
\begin{split}
\left\langle r^2\right\rangle(\omega+2) & = \varepsilon^2 \left\{ (\omega+2)^2\left[\frac{1+g(\omega+1)}{2}\right]\left[\frac{1+g(\omega+2)}{2}\right] \right. \\ & +\omega^2 \left[\frac{1+g(\omega+1)}{2}\right]\left[\frac{1-g(\omega+2)}{2}\right] \\ &
+ \omega^2 \left[\frac{1-g(\omega+1)}{2}\right]\left[\frac{1+g(\omega+2)}{2}\right] \\ & \left. +(\omega-2)^2\left[\frac{1-g(\omega+1)}{2}\right]\left[\frac{1-g(\omega+2)}{2}\right] \right\} \\ &
\frac{\varepsilon^2}{4} \left\{ (\omega^2+4\omega+4)[1+g(\omega+1)+g(\omega+2)+g(\omega+1)g(\omega+2)]\right. \\ &
+ \omega^2[1+g(\omega+1)-g(\omega+2)-g(\omega+1)g(\omega+2)] \\ &
+\omega^2[1-g(\omega+1)+g(\omega+2)-g(\omega+1)g(\omega+2)] \\ &
\left. + (\omega^2-4\omega+4)[1-g(\omega+1)-g(\omega+2)+g(\omega+1)g(\omega+2)] \right\},
\end{split}
\end{equation*}
which reduces to:
\begin{equation}
\left\langle r^2\right\rangle(\omega+2)=\varepsilon^2\left\{ \omega^2+2+2g(\omega+1)g(\omega+2)+2\omega[g(\omega+1)+g(\omega+2)] \right\}.
\end{equation}
\paragraph{Any $k$}
We can proceed for further $k$ and will arrive at the following
expression for any $k>\omega$:
\begin{equation}
\begin{split}
\left\langle r^2\right\rangle (\omega+k) & =\varepsilon^2\left[\omega^2+k+2\sum_{i=1}^k\sum_{j=i}^kg(\omega+i)g(\omega+j)\right. \\ & \left.-2\sum_{i=1}^kg^2(\omega+1)+2\omega\sum_{i=1}^kg(\omega+i)\right]
\end{split}
\end{equation}
which by using the definition of the diffusion coefficient and particle speed can be converted to
\begin{equation*}
\begin{split}
\left\langle r^2\right\rangle(\omega+k) & = 2dD\left[k\tau-2\sum_{i=1}^kg^2(\omega+i)\tau \right]+v^2\left[2\sum_{i=1}^k\sum_{j=i}^kg(\omega+i)g(\omega+j)\tau^2\right. \\ &
\left. +(\omega \tau)^2+2\omega \tau \sum_{i=1}^kg(\omega+i)\tau \right]
\end{split}
\end{equation*}
which in the limit $\tau \rightarrow 0$ is
\begin{equation}
\begin{split}
\left\langle r^2\right\rangle(t) & =2dD\left[(t-t^{\star})-2\int_{t^{\star}}^tg^2(\tau)\mathrm{d}\tau \right] \\ &
+v^2\left[2\int_{t^{\star}}^t\int_{\tau}^tg(\tau)g(k)\mathrm{d}k\mathrm{d}\tau + t^{\star 2}+ 2t^{\star}\int_{t^{\star}}^tg(\tau)\mathrm{d}\tau\right].
\end{split}
\label{msdplus}
\end{equation}
Combining Supplementary Eqs.~(\ref{msdminus}) and (\ref{msdplus}) we obtain the MSD
of a particle with a piecewise power-law decaying VACF:
\begin{equation}
\left\langle r^2\right\rangle(t) = \begin{cases}(vt)^2 & t\leq t^{\star}, \\
\!\begin{aligned}
& 2dD\left[(t-t^{\star})-2\int_{t^{\star}}^tg^2(\tau)\mathrm{d}\tau \right] \\ &
+v^2\left[2\int_{t^{\star}}^t\int_{\tau}^tg(\tau)g(k)\mathrm{d}k\mathrm{d}\tau + t^{\star 2}+ 2t^{\star}\int_{t^{\star}}^tg(\tau)\mathrm{d}\tau\right]
\end{aligned} & t>t^{\star}.
\end{cases}
\end{equation}

\section{Generalized time-correlated random walk: rule derivation}
We maximize the caliber
\begin{equation}
 \mathcal{C}=-\sum_{\Gamma}P_{\Gamma}\ln P_{\Gamma},
\end{equation}
subject to observing a certain VACF, which translates into the
Lagrange multiplier problem
\begin{equation}
 \tilde{\mathcal{C}}\left[P_{\Gamma}\right]=-\sum_{\Gamma}P_{\Gamma}\ln P_{\Gamma}+\sum_{i=1}^k \beta(i) \left[ \sum_{\Gamma}P_{\Gamma}\left(\vec{c}_{n_0} \cdot \vec{c}_{n_i}\right)-g(i)\right]+\lambda\left(\sum_{\Gamma}P_{\Gamma}-1\right),
\end{equation}
This yields the trajectory probabilites
\begin{equation}
 P_{\Gamma}=\frac{1}{Z}\exp\left[\sum_{i=1}^k\beta(i)\left(\vec{c}_{n_0} \cdot \vec{c}_{n_i}\right)\right],
 \label{pGamm}
\end{equation}
where $Z=\exp\left(1-\lambda\right)$ is called the dynamical partition
function which, by optimizing the functional with respect to $\lambda$
(i.e., $\frac{\partial \tilde{\mathcal{C}}}{\partial \lambda}=0$), is
given by $
Z=\sum_{\Gamma}\exp\left[\sum_{i=1}^k\beta(i)\left(\vec{c}_{n_0} \cdot
    \vec{c}_{n_i}\right)\right].  $ Optimizing with respect to $\beta(i)$
yields our original constraint
\begin{equation}
 g(k)=\sum_{\Gamma}P_{\Gamma}\left(\vec{c}_{n_0} \cdot \vec{c}_{n_i}\right).
 \label{vacf_lgca2}
\end{equation}
Solving for $\beta(i)$ using Supplementary Eqs.~(\ref{pGamm}) and (\ref{vacf_lgca2}) can
be quite challenging, so we expand $P_{\Gamma}$ in a Taylor series
around $\beta(i)=0$, which reduces to $ g(k)\approx
\sum_{\Gamma}\frac{1}{Z}\left[1+\sum_{i=0}^k\beta(i)\left(\vec{c}_{n_0}
    \cdot \vec{c}_{n_i}\right)\right]\left(\vec{c}_{n_0} \cdot
  \vec{c}_{n_k}\right), $ where the dynamical partition function is
simplified as
\begin{equation*}
 \begin{split}
  Z & \approx\sum_{\Gamma}\left[1+\sum_{i=0}^k\beta(i)\left(\vec{c}_{n_0} \cdot \vec{c}_{n_i}\right)\right]=b^k+\sum_{\Gamma}\sum_{i=1}^k\beta(i)\left(\vec{c}_{n_0} \cdot \vec{c}_{n_i}\right)=b^k+\sum_{i=1}^k\beta(i)\sum_{\Gamma}\left(\vec{c}_{n_0} \cdot \vec{c}_{n_i}\right) \\&
  =b^k+\sum_{i=1}^k\beta(i)\sum_{\Gamma}\cos\theta_i=b^k,
 \end{split}
\end{equation*}
where $b$ is the number of lattice directions and $\theta_i$ is the
angle between the original particle orientation and the particle orientation
at time step $i$.  $\sum_{\Gamma}\cos\theta_i=0$ because the lattice
directions and hence the possible particle orientations are
symmetrically and homogeneously distributed.  Substituting $Z$, using
the same notation as previously, employing trigonometric identities,
and denoting the spatial dimension by $d$, we proceed with the
calculation:
\begin{equation*}
 \begin{split}
  g(k) & \approx b^{-k}\left[\sum_{\Gamma}\cos \theta_k + \sum_{\Gamma}\sum_{i=1}^k\beta(i)\cos \theta_i \cos \theta_k\right]=b^{-k}\sum_{\Gamma}\sum_{i=1}^k\beta(i)\cos \theta_i \cos \theta_k =b^{-k}\sum_{i=1}^k\sum_{\Gamma}\beta(i)\cos \theta_i \cos \theta_k \\
  & =b^{-k}\left[\beta(1)\sum_{\Gamma}\cos \theta_1 \cos \theta_k+\beta(2)\sum_{\Gamma}\cos \theta_2 \cos \theta_k+ \cdots +\beta(k-1)\sum_{\Gamma}\cos \theta_{k-1} \cos \theta_k+\beta(k)\sum_{\Gamma} \cos^2 \theta_k\right] \\
  & =b^{-k}\beta(k)\sum_{\Gamma} \cos^2 \theta_k=\frac{\beta(k)}{2b^k}\sum_{\Gamma}\left[1+\cos\left(2\theta_k\right)\right]=\frac{\beta(k)}{2b^k}\left[b^k+\sum_{\Gamma}\cos\left(2\theta_k\right)\right]=\frac{\beta(k)}{2b^k}\left[b^k+\frac{b^k}{d}\left(2-d\right)\right] \\
  & =\frac{\beta(k)}{2}\left[ 1+\frac{2-d}{d} \right]=\frac{\beta(k)}{d},
 \end{split}
\end{equation*}
which determines the Lagrange multiplier
\begin{equation}
 \beta(k)=dg(k).
\end{equation}
So the generalized probabilities are finally
\begin{equation}
 P_{\Gamma}=\frac{1}{Z}\exp\left[\sum_{i=1}^kdg(i)\left(\vec{c}_{n_0} \cdot \vec{c}_{n_i}\right)\right],
\end{equation}
which is the probability for the whole trajectory. Due to the
exponential form of this probability, we can decompose the
trajectory probability into reorientation probabilities:
\begin{equation}
 P_{\Gamma}=\prod_{i=1}^kP_{x,k},
\end{equation}
given by:
\begin{equation}
 P_{n_k,k}=\frac{1}{z}\exp\left[dg(k)\left(\vec{c}_{n_0} \cdot \vec{c}_{n_k}\right)\right],
\end{equation}
where $z$ is the normalization constant for the reorientation probability.

\section{Generalized time-correlated random walk: VACF decay analysis}
Eq.~(\ref{decaytanpow}) is at first sight, different from a simple
power law decay. We now assess how similar Eq.~(\ref{decaytanpow}) is to
a simple power law decay for intermediate times.  The easiest and most
insightful way to achieve this is to expand both Eq.~(\ref{decaytanpow})
and a generic power law in a Taylor series, and to compare the Taylor
coefficients. We expand around $t=\Delta$.  We will denote
Eq.~(\ref{decaytanpow}) by $C(t)$. The power law function has the
following form:
\begin{equation}
 G(t)=G_1\left(\frac{\Delta}{t}\right)^{\gamma},
\label{powapprox}
\end{equation}
where the constants $G_1$ and $\gamma$ are unspecified.  First, we
calculate the first two derivatives of $C(t)$:
\begin{subequations}
 \begin{align}
  \frac{\mathrm{d}C(t)}{\mathrm{d}t}=-\phi C_0\Delta^{\phi} t^{-\phi-1}\sech^2\left[C_0\left( \frac{\Delta}{t} \right)^{\phi}\right] & \\
   \frac{\mathrm{d}^2C(t)}{\mathrm{d}t^2}=\phi C_0 \Delta^{\phi}t^{-2\phi-2}\sech^2\left[ C_0\left( \frac{\Delta}{t} \right)^{\phi} \right]\left\{ (\phi+1)t^{\phi}-2\phi C_0 \Delta^{\phi} \tanh \left[C_0\left( \frac{\Delta}{t} \right)^{\phi}\right]\right\},
 \end{align}
\end{subequations}
with which we can calculate its Taylor series up to the second order
term:
\begin{equation}
\begin{split}
 C(t) & =\tanh\left(C_0\right)-\frac{\phi C_0}{\Delta}\sech^2\left(C_0\right)(t-\Delta)+\\ & \frac{1}{2!}\frac{\phi C_0 \sech^2\left(C_0\right)}{\Delta^2}\left\{1+\phi\left[1-2C_0\tanh\left(C_0\right)\right]\right\}(t-\Delta)^2+\mathcal{O}(t^3).
 \end{split}
 \label{taytanh}
\end{equation}
We now proceed in the same way with the power law:
\begin{subequations}
 \begin{align}
  \frac{\mathrm{d}G(t)}{\mathrm{d}t}=-G_1\Delta^\gamma \gamma t^{-\gamma-1} & \\
  \frac{\mathrm{d}^2G(t)}{\mathrm{d}t^2}=G_1\Delta^\gamma \gamma(1+\gamma)t^{-\gamma-2}
 \end{align}
\end{subequations}
and expand in a Taylor series around $t=\Delta$:
\begin{equation}
 G(t)=G_1-\frac{G_1\gamma}{\Delta} (t-\Delta)+\frac{1}{2!}\frac{G_1\gamma}{\Delta^2}(1+\gamma)(t-\Delta)^2+\mathcal{O}(t^3).
 \label{pretaypow}
\end{equation}
To determine $G_1$ and $\gamma$ we equate the zeroth and first order terms of Supplementary Eqs.~(\ref{taytanh}) and (\ref{pretaypow}), which yields
\begin{subequations}
 \begin{align}
  G_1=\tanh\left(C_0\right) \label{g1} & \\ 
  \gamma=\phi C_0\frac{\sech^2\left(C_0\right)}{\tanh\left(C_0\right)} \label{gamma} ,
 \end{align}
\end{subequations}
so that the Taylor series expansion is determined by
\begin{equation}
\begin{split}
 G(t) & =\tanh\left(C_0\right)-\frac{\phi C_0}{\Delta} \sech^2\left(C_0\right)(t-\Delta)+ \\ & \frac{1}{2!}\frac{\phi C_0 \sech^2\left(C_0\right)}{\Delta^2}\left[1+\phi C_0 \frac{\sech^2\left(C_0\right)}{\tanh\left(C_0\right)}\right](t-\Delta)^2+\mathcal{O}(t^3).
\end{split}
\label{taypow}
 \end{equation}
 To estimate the similarity between both decays, we calculate the
 difference between Supplementary Eqs.~(\ref{taytanh}) and (\ref{taypow}) up to
 second order terms:
\begin{equation*}
 \begin{split}
  C(t)-G(t) & \approx \frac{1}{2!}\frac{\phi C_0 \sech^2\left(C_0\right)}{\Delta^2}\left\{1+\phi\left[1-2C_0 \tanh\left(C_0\right)\right]\right\}(t-\Delta)^2 - \\ & \frac{1}{2!}\frac{\phi C_0\sech^2\left(C_0\right)}{\Delta^2}\left[1+\phi C_0 \frac{\sech^2\left(C_0\right)}{\tanh\left(C_0\right)}\right](t-\Delta)^2 = (t-\Delta)^2\cdot \\ &
  \frac{1}{2!}\frac{\phi C_0\sech^2\left(C_0\right)}{\Delta^2}\left\{ 1+\phi\left[1-2C_0\tanh\left(C_0\right)\right] - 1-\phi C_0\frac{\sech^2\left(C_0\right)}{\tanh\left(C_0\right)} \right\} \\ &
  =(t-\Delta)^2\frac{1}{2!}\frac{\phi^2C_0\sech^2\left(C_0\right)}{\Delta^2}\left[ 1-2C_0\tanh\left(C_0\right)-C_0\frac{\sech^2\left(C_0\right)}{\tanh\left(C_0\right)} \right] \\ &
  =(t-\Delta)^2\frac{1}{2!}\frac{\phi^2C_0\sech^2\left(C_0\right)}{\Delta^2}\left\{ 1-C_0\left[\frac{2\sinh(C_0)}{\cosh(C_0)}-\frac{1}{\cosh(C_0)\sinh(C_0)} \right]\right\} \\ &
  =(t-\Delta)^2\frac{1}{2!}\frac{\phi^2C_0\sech^2\left(C_0\right)}{\Delta^2}\left\{ 1-C_0\left[\frac{2\sinh^2(C_0)+1}{\cosh(C_0)\sinh(C_0)}\right] \right\}
 \end{split}
\end{equation*}
which, after using hyperbolic identities, can be simplified to
\begin{equation}
 C(t)-G(t) \approx (t-\Delta)^2\frac{1}{2!}\frac{\phi^2C_0\sech^2\left(C_0\right)}{\Delta^2}\left[ 1-2C_0\coth\left(2C_0\right) \right] \propto \left(\frac{\phi}{\Delta}\right)^2.
 \label{difftay}
\end{equation}

\section{LGCA simulations}
\subsection*{Persistent random walk}
Simulations were performed with only one particle with the reorientation 
probability given by Eq.~(\ref{lgcatrans}) whose displacement and
orientation were tracked at every time step. The lattice spacing was
set to $\varepsilon=0.25$, and the time step to $\tau=0.015625$.  The
total simulation consisted of 100 time steps. The sensitivity (related
to the internal force required for reorientation) was varied from
$\beta=3$ to $\beta=5$.  Simulations were repeated 1000 times for each
sensitivity in order to obtain statistically relevant
results. Simulation results for low and high sensitivities are shown
in Fig.~\ref{theovssims}i.

As expected, correlations die off more slowly with increasing
sensitivity. On the other hand, the MSD quickly starts behaving
linearly, except for times close to zero, where it behaves almost
ballistically. The region where displacement is almost ballistic
increases with increasing sensitivity.

Additionally, we observe that the derived continuous time expressions
agree perfectly with the discrete LGCA simulations.

\subsection*{Time-correlated random walk}
Simulations were performed with only one particle with the reorientation
probability given by Eq.~(\ref{pxgeneral}). The lattice spacing was
set to $\varepsilon=0.25$, and the time step to $\tau=0.015625$. The
constant $C_0$ was set to $0.5$, and the crossover time was equal to
the time step length, $\Delta=\tau=0.015625$.  The total simulation consisted of 1000 time
steps. Three different exponents were evaluated: $\phi=0.1$, $\phi=1$,
and $\phi=9$.  Simulations were repeated 1000 times for each exponent,
in order to obtain statistically relevant results. Simulation results
for small and large exponents are shown in Fig.~\ref{theovssims}ii, as well as
a plot of Eqs.~(\ref{anom_corr}) and~(\ref{msdgenpow}) (integrated with
MATLAB). We can see that Eqs.~(\ref{anom_corr}) and~(\ref{msdgenpow})
match the simulation data perfectly. We also observe that for low values of the exponent
$\phi$ the particle moves superdiffusively while for large values the
particle diffuses normally.

\subsection*{Generalized time-correlated random walk}
Simulations were performed with only one particle with probabilities given
by Eq.~(\ref{pxgeneral}). The lattice spacing was set to
$\varepsilon=0.25$ and the time step to $\tau=0.015625$. The constant
$C_0$ was set to $0.5$ and the crossover time was equal to
the time step length, $\Delta=\tau=0.015625$.
 The total simulation consisted of 100 time steps. Two
different exponents were evaluated, $\phi=0.1$ and $\phi=1$.
Simulations were repeated 1000 times for each exponent, in order to
obtain statistically relevant results. Simulation results for small,
and large exponents are shown in Fig.~\ref{theovssims}iii as well as a plot of
Eqs.~(\ref{decaytanpow}) and~(\ref{contmsdexp}) (integrated with
MATLAB). We can see that Eq.~(\ref{decaytanpow}) and~(\ref{contmsdexp})
match the simulation data perfectly.  Comparing Figs.~\ref{theovssims}ii and
\ref{theovssims}iii, it is evident that the VACF in both cases is quite similar, as
expected given the small value of $\Delta$ used in these simulations.
 We also observe that for low values of the exponent
$\phi$ the particle moves superdiffusively while for large values the
particle diffuses normally.

\end{appendix}
\end{document}